\documentclass[11pt,english]{article}
\usepackage{mathptmx}
\usepackage[T1]{fontenc}
\usepackage[latin9]{inputenc}
\usepackage{float}
\usepackage{mathtools}
\usepackage{enumitem}
\usepackage{amsmath}
\usepackage{amsthm}
\usepackage{amssymb}
\usepackage[a4paper]{geometry}
\usepackage{setspace}
\usepackage[authoryear]{natbib}
\makeatletter
\newlength{\lyxlabelwidth}      
	\newenvironment{elabeling}[2][]%
	{\settowidth{\lyxlabelwidth}{#2}
		\begin{description}[font=\normalfont,style=sameline,
			leftmargin=\lyxlabelwidth,#1]}
	{\end{description}}
\theoremstyle{definition}
\newtheorem{defn}{\protect\definitionname}
\theoremstyle{remark}
\newtheorem*{rem*}{\protect\remarkname}
\theoremstyle{plain}
\newtheorem{lem}{\protect\lemmaname}
\theoremstyle{plain}
\newtheorem{thm}{\protect\theoremname}
\theoremstyle{plain}
\newtheorem{cor}{\protect\corollaryname}
\theoremstyle{plain}
\newtheorem{prop}{\protect\propositionname}

\usepackage{algorithm,algpseudocode}

\usepackage{caption}

\numberwithin{equation}{section}

\newcommand*\dif{\mathop{}\!\mathrm{d}}

\usepackage{enumitem}  
\setlist[itemize]{noitemsep,topsep=2pt}

\setlist{noitemsep,topsep=2pt} 

\usepackage[font=small,labelfont=bf]{caption}

\usepackage{enumitem}
\setlist[enumerate,1]{label=(\alph*),font=\normalfont}

\makeatother

\usepackage{babel}
\providecommand{\corollaryname}{Corollary}
\providecommand{\definitionname}{Definition}
\providecommand{\lemmaname}{Lemma}
\providecommand{\propositionname}{Proposition}
\providecommand{\remarkname}{Remark}
\providecommand{\theoremname}{Theorem}

\begin{document}
\title{Sequential Portfolio Selection under Latent Side Information-Dependence
Structure: Optimality and Universal Learning Algorithms}
\author{Duy Khanh Lam\thanks{Some of the main results of this paper were presented in a poster
titled ``\emph{Universal Investment Strategies by Online Learning
under Latent Information-Dependence Structure}'' at the workshop
on ``\emph{Learning and Inference from Structured Data: Universality,
Correlations, and Beyond}'' held at the Abdus Salam International
Centre for Theoretical Physics in Italy, from July 3 to 7, 2023. I
would also like to express my gratitude to my advisor, Giulio Bottazzi,
and to Daniele Giachini at the Scuola Superiore Sant'Anna for their
valuable discussions.}\\Scuola Normale Superiore\\ \\(Working paper)}
\maketitle
\begin{abstract}
This paper investigates the investment problem of constructing an
optimal no-short sequential portfolio strategy in a market with a
latent dependence structure between asset prices and partly unobservable
side information, which is often high-dimensional. The results demonstrate
that a dynamic strategy, which forms a portfolio based on perfect
knowledge of the dependence structure and full market information
over time, may not grow at a higher rate infinitely often than a constant
strategy, which remains invariant over time. Specifically, if the
market is stationary, implying that the dependence structure is statistically
stable, the growth rate of an optimal dynamic strategy, utilizing
the maximum capacity of the entire market information, almost surely
decays over time into an equilibrium state, asymptotically converging
to the growth rate of a constant strategy.\smallskip

Technically, this work reassesses the common belief that a constant
strategy only attains the optimal limiting growth rate of dynamic
strategies when the market process is identically and independently
distributed. By analyzing the dynamic log-optimal portfolio strategy
as the optimal benchmark in a stationary market with side information,
we show that a random optimal constant strategy almost surely exists,
even when a limiting growth rate for the dynamic strategy does not.
Consequently, two approaches to learning algorithms for portfolio
construction are discussed, demonstrating the safety of removing side
information from the learning process while still guaranteeing an
asymptotic growth rate comparable to that of the optimal dynamic strategy.\smallskip
\end{abstract}
{\small\textit{Keywords}}{\small : Online Learning, Dynamic Strategy,
Universality, Optimal Growth, Side Information, Latent Dependence
Structure.}{\small\par}

\pagebreak\setlength{\abovedisplayskip}{2.5pt} 
\setlength{\abovedisplayshortskip}{2.5pt}
\setlength{\belowdisplayskip}{2.5pt} 
\setlength{\belowdisplayshortskip}{2.5pt} 

\section{Introduction}

In sequential investment, the task of collecting and analyzing data
on market events that correlate with future assets' prices, to make
well-informed decisions at each investment period, is undoubtedly
one of the most significant concerns. In practice, a market is often
represented by latent dynamics, designed to capture the hidden dependence
structures among assets' prices, their historical data, and other
sources of side market information, such as energy prices, news, and
other potentially correlated features. However, investors consistently
face three fundamental challenges of uncertainty: the incompleteness
of observable side market information, the computational difficulty
associated with high-dimensional data, and the unknown nature of the
market\textquoteright s dependence structure. The first challenge
stems from the fact that valuable side market information is often
costly and usually available only to a small subset of investors at
any point in time, such as insiders, which makes most of the side
information unobservable. The second challenge lies in the uncertainty
and potentially large volume of market features correlated with future
assets' prices, making estimation from collected data computationally
expensive in terms of both time and processing power, which complicates
high-frequency decision-making. Lastly, the third challenge arises
from the hidden dependence structure of the market, where the statistical
properties of its dynamics are unknown, leaving open the possibility
that the process could either be identically and independently distributed
(i.i.d.) or exhibit an undetermined memory length.\smallskip

\textbf{Motivation and overall result}. In this paper, we explore
the problem of constructing sequential no-short portfolios and address
the three aforementioned challenges based on a general model of a
market process with partially observable side market information.
We demonstrate that the best dynamic strategy, defined as a sequence
of portfolios determined using the complete history of market observations
over time, including even the unobservable features, cannot infinitely
often outperform an optimal constant strategy, represented as a unique
and random portfolio, in terms of both growth rate and expected growth
rate as time progresses to infinity. This finding implies that the
constant strategy does not require knowledge of periodic side market
information, yet it remains asymptotically optimal. Such a result
contradicts conventional perspectives in the fields of information
theory, finance, and potentially learning theory, which suggest that
possessing full knowledge of market information and dependence structures
should unequivocally outperform ignorance. Based on our findings,
a class of learning algorithms that relies solely on past assets'
prices, without any prior knowledge of latent dependence structures
or side information, is shown to be optimal and comparable to another
innovative learning algorithm that is more sophisticated and leverages
complete market knowledge. This analysis underscores the safety of
omitting side feature data in the learning process while still ensuring
optimality in the asymptotic growth rate.\smallskip

\textbf{Paper organization and main novelties}. We briefly describe
the main novel results, which contradict popular conjectures and the
status quo, along with the analysis techniques and literature references
for comparison:\smallskip
\begin{elabeling}{00.00.0000}
\item [{\textit{Section~2}:}] The assets' returns and all side information,
including unobservable features, are modeled as a pairwise process
to capture various dependence structures. Then, the log-optimal strategy
is introduced as the optimal benchmark among all strategies. Moreover,
it is defined without requiring the integrability assumption for the
variables, broadening the scope of market investigations. Note that
integrability is a common assumption in the literature, ensuring that
expectations are well-defined.$\vspace{0.5ex}$
\item [{\textit{Section~3}:}] The optimality properties of the log-optimal
strategy are established as a generalization of \citet{Algoet1988}.
The analysis is performed under a measure transformation, which bypasses
the required integrability by decomposing the portfolio's return into
an integrable part. The properties demonstrate the advantages of utilizing
full side market information and guarantee the optimal expected growth
rate and asymptotic growth rate for the log-optimal strategy over
other competitive strategies using less information.$\vspace{0.5ex}$
\item [{\textit{Section~}4.1:}] This section reassesses the status quo
in contemporary literature. Since the works of \citet{Kelly1956,Breiman1960,Breiman1961},
it has been widely accepted that a constant strategy of a single portfolio
only achieves the optimal asymptotic growth rate under an i.i.d. process
of assets' returns without side information, as reiterated in many
later articles such as \citet{Cover1991,Morvai1991,Algoet1992,Morvai1992,Cover1996,Ordentlich1996,Cover1996a,Gyorfi2006,Cover2006,CesaBianchi2006,Bhatt2023},
etc. Additionally, after demonstrating the optimal limiting growth
rate of the log-optimal strategy under stationary and ergodic processes
without side information in \citet{Algoet1988}; the author conjectured
that no constant strategy can match the optimal growth rate of a dynamic
strategy under the same process in \citet{Algoet1992}\footnote{See the last paragraph in Remark 4, pages 929-930, where the author
also restates that a constant strategy attains the optimal asymptotic
growth rate only under an i.i.d. process without side market information.}. However, we show that a random optimal constant strategy can almost
surely asymptotically approach the optimal growth rate of the log-optimal
strategy under a stationary market process with side information,
becoming non-random if the process is ergodic.$\vspace{0.5ex}$
\item [{\textit{Section~}4.2:}] This section establishes two equalities
between the optimal limits of the growth rate and expected growth
rate of the log-optimal strategy utilizing the infinitely large amount
of market information, and those of the random optimal strategy, which
is not conditioned on side information. It implies that under the
stationarity condition, the advantages of utilizing the maximum capacity
of market information will almost surely decay over time into an equilibrium
state.$\vspace{0.5ex}$
\item [{\textit{Section~}4.3:}] The general process is reduced to capture
several models of market phenomena, including a market with a non-decaying
impact of past information on future market, a market with finite-order
memory, and one where future assets' returns depend only on side information
but not on the past. Generally, if the latent dependence structure
remains unchanged over time, a random optimal constant strategy could
exist.$\vspace{0.5ex}$
\item [{\textit{Section~}5:}] Under an unknown distribution, making the
optimal strategies undetermined in advance, two approaches for universal
learning algorithms to construct strategies are discussed. The first
approach requires learning only past assets' returns to asymptotically
attain the same rate as the random optimal constant strategy. This
generalizes the optimality of the Universal portfolio by \citet{Cover1991}
beyond an i.i.d. process without side information. Alternatively,
the second approach requires learning the entire past market information
and dependence structure to asymptotically attain the same rate as
the log-optimal strategy. This is based on approaches proposed by
\citet{Algoet1992} and \citet{Gyorfi2006}, but it can learn less
information while still guaranteeing optimality.$\vspace{0.5ex}$
\end{elabeling}

\section{Model settings and formalizations}

We consider modeling an investment scenario where an investor sequentially
constructs no-short portfolios in a multi-asset market without frictions.
Over time, the investor observes market information, including historical
asset prices and side information such as oil prices, news, and events
from other markets. While future asset prices remain inherently unpredictable,
the revealed information can be utilized in various ways to develop
investment strategies aimed at achieving favorable cumulative wealth
across diverse market possibilities. However, not all side information
correlated with asset prices is observable, and its relevance varies.
Additionally, collecting extensive side information can be prohibitively
expensive. Thus, an ideal model should be universal, capable of accommodating
diverse dependence structures and incorporating unobservable side
information, even when such information is infinite in scope. This
universality naturally introduces the challenge of high dimensionality.
Furthermore, the model should facilitate the formulation of an optimal
strategy from the set of all accessible strategies, serving as a benchmark
despite the uncertainties of future market behavior and incomplete
information. Such a desired modeling framework, which addresses these
challenges, is presented as follows.\smallskip

Let us consider sequential portfolio construction for investment in
a stock market with $m\geq2$ risky assets, whose prices are causally
affected by $k\geq1$ features of observable and unobservable side
information, over discrete time periods $n\in\mathbb{N}_{+}$. Assume
that the assets' returns and side information are realized with respect
to a pairwise random process $\big\{ X_{n},Y_{n}\big\}_{n=1}^{\infty}$
that is jointly defined on the canonical probability space $\text{\ensuremath{\big(\Omega},\ensuremath{\mathbb{F}},\ensuremath{\mathbb{\mathbb{P}}\big)}}$,
i.e., $\big(X_{n},Y_{n}\big)\big(\omega\big)=\big(X_{n}\big(\omega\big),Y_{n}\big(\omega\big)\big)$
for all $n$, where $\omega\in\Omega$. Specifically, the positive
real-valued vector $X_{n}\coloneqq\big(X_{n,1},...,X_{n,m}\big)\in\mathbb{R}_{++}^{m}$
represents the positive returns of the assets at time $n$, which
is defined as the ratio of the assets' prices at time $n$ to those
at the previous time $n-1$; while the real-valued random vector $Y_{n}\coloneqq\big(Y_{n,1},...,X_{n,k}\big)\in\mathbb{R}^{k}$
represents the side information related to the market's features at
time $n$. We further denote $X_{1}^{n}\coloneqq\left\{ X_{1},...,X_{n}\right\} $
and $Y_{1}^{n}\coloneqq\left\{ Y_{1},...,Y_{n}\right\} $ as the finite
sequences of asset returns and side information variables, respectively,
from $1$ to $n$. It should be remarked that the side information
variables $Y_{n}$ could take values in some Polish spaces rather
than $\mathbb{R}^{k}$ as the default setting.\smallskip

At each time $n$, the sub-$\sigma$-field $\mathcal{F}_{n}\coloneqq\sigma\big(X_{1}^{n},Y_{1}^{n}\big)\subseteq\mathbb{F}$
embodies all past information up to the present of the random pairs
$\big(X_{n},Y_{n}\big)$, up to the limiting $\sigma$-field $\mathcal{F}_{\infty}\coloneqq\sigma\big(\cup_{n=1}^{\infty}\mathcal{F}_{n}\big)\subseteq\mathbb{F}$,
which encompasses the infinite past as $n\to\infty$ and is the smallest
$\sigma$-field containing the largest possible information the investor
can know. Consequently, it is possible to define a further sub-$\sigma$-field
$\mathcal{F}_{n}^{X}=\sigma\big(X_{1}^{n-1},Y_{1}^{n}\big)\subset\mathcal{F}_{n+1}$,
which includes only the side information just before the realization
of the assets' returns at each time $n$. Thus, the corresponding
limiting information field generated by $\mathcal{F}_{n}^{X}$, as
an infinite union, is denoted by the $\sigma$-field $\mathcal{F}_{\infty}^{X}\coloneqq\sigma\big(\cup_{n=1}^{\infty}\mathcal{F}_{n}^{X}\big)\subseteq\mathbb{F}$.
Noting that although we could model the side information events that
occur between two consecutive $X_{n}$ and $X_{n+1}$ as $Y_{n}$
instead of $Y_{n+1}$, which would make the market process simpler,
it seems difficult to model a scenario where the present assets' returns
depend solely on new market events but are independent of past ones.
Hence, our modeling framework allows us to conveniently capture various
market behaviors, as discussed in the section on investigating the
stationary market.\smallskip

For each $n$, given the $\sigma$-field $\mathcal{F}_{n}^{X}$, a
portfolio $b_{n}:\mathbb{R}_{++}^{m\times(n-1)}\times\mathbb{R}^{k\times n}\mathbb{\rightarrow}\mathcal{B}^{m}$,
where the simplex $\mathcal{B}^{m}\coloneqq\big\{\beta\coloneqq\big(\beta_{1},...,\beta_{m}\big)\in\mathbb{R}_{++}^{m}:\,{\displaystyle {\textstyle \sum}}_{{\scriptstyle j=1}}^{{\scriptstyle m}}\beta_{j}=1\big\}$
denotes the domain of all no-short-constrained portfolios. This means
each portfolio is causally selected based on the realizations of $X_{1}^{n-1}$
and $Y_{1}^{n}$, but not $X_{1}^{n}$, in order to avoid the unrealistic
and impossible case of the future assets' returns being always predicted
perfectly. We subsequently denote the corresponding \emph{strategy},
formed from the portfolios $b_{n}$ for all $n$, as the infinite
sequence $\big(b_{n}\big)\coloneqq\left\{ b_{n}\right\} _{n=1}^{\infty}$.
Accordingly, if a strategy uses only a fixed portfolio over time,
such that $b_{n}\coloneqq b$ for all $n$, it is termed a \emph{constant
strategy} and denoted as $\big(b\big)$ without the time index. Additionally,
the return of a portfolio $b$ with respect to assets' returns $X_{n}$
is denoted by $\langle b,X_{n}\rangle$, where $\langle\cdot,\cdot\rangle$
represents the scalar product of two vectors.\smallskip

With the above settings, let the initial capital $S_{0}\big(b_{0}\big)\eqqcolon S_{0}=1$
by convention for any strategy, and assume that the portfolios are
constructed without commission fees for arbitrary fractions. The \emph{cumulative
wealth} and its corresponding exponential average \emph{growth rate}
after $n$ periods of investment, yielded by a self-financed strategy
$\big(b_{n}\big)$, are respectively defined as follows:
\[
S_{n}\big(b_{n}\big)\coloneqq{\displaystyle {\displaystyle \prod_{i=1}^{n}}\left\langle b_{i},X_{i}\right\rangle }\text{ and }W_{n}\big(b_{n}\big)\coloneqq\dfrac{1}{n}\log S_{n}\big(b_{n}\big)={\displaystyle \dfrac{1}{n}\sum_{i=1}^{n}\log\left\langle b_{i},X_{i}\right\rangle ,\,\forall n,}
\]
which shorthand for $S_{n}\big(\left\{ b_{i}\right\} _{i=1}^{n},X_{1}^{n}\big)$
and $W_{n}\big(\left\{ b_{i}\right\} _{i=1}^{n},X_{1}^{n}\big)$,
respectively. It is worth noting that the assumption of positive assets'
returns is not only realistic but also serves to exclude the case
where a strategy's cumulative wealth is depleted to zero, which would
halt the investment.\smallskip

\subsection{The benchmark strategy and normalized assets' returns}

In general, to investigate the intrinsic growth rate of a market,
it is sufficient to analyze a single representative strategy that
is guaranteed not to yield a lower growth rate than any other strategy
as time evolves. Therefore, any potentially competitive strategy must
be compared solely with this representative strategy for performance
benchmarking. The so-called \emph{log-optimal strategy}, as defined
in Definition \ref{Definition 1}, provides such a strategy for the
purpose of our investigation. However, to better understand the nature
of the market under consideration, it should be noted that the benchmarking
role of the log-optimal strategy may no longer hold when market frictions,
such as commission fees, are introduced.
\begin{defn}
Given the filtration $\left\{ \mathcal{F}_{n}^{X}\right\} _{n=1}^{\infty}$,
consider the portfolios $b_{n}^{*}$ over periods $n$ that satisfy
the following condition:\label{Definition 1}
\[
\mathbb{E}\big(\log\left\langle b,X_{n}\right\rangle -\log\left\langle b_{n}^{*},X_{n}\right\rangle |\mathcal{F}_{n}^{X}\big)\leq0,\,\forall b\in\mathcal{B}^{m},\forall n.
\]
Any sequence of such $b_{n}^{*}$ defines a log-optimal strategy and
is denoted as $\big(b_{n}^{*}\big)$ henceforth.
\end{defn}
We stress that Definition \ref{Definition 1} implies the log-optimal
portfolios are ones maximizing the conditional expectation of the
differences between logarithmic returns of two portfolios, given all
past market assets' returns and side information up to the present
period, regardless of their observability. This way of formalization
does not necessitate the finite existence of individually conditional
expected values of each portfolio's logarithmic returns but their
well-defined conditional expected difference. Therefore, it allows
us to broaden our investigation of a stochastic market to the general
case of infinite or even ill-defined expectations by bypassing the
required integrability of measurable variables, which will be discussed
in the next sections.\smallskip
\begin{rem*}
In the literature, the integrability of $X_{n}$, i.e., the expected
values $\mathbb{E}\big(\log\left\langle b,X_{n}\right\rangle \big)$,
is often assumed to finitely exist in some form as a minimal requirement
for analysis in several papers, such as \citet{Morvai1991,Cuchiero2019,Gyorfi2006},
etc. In the paper \citet{Algoet1988}, the definition of the log-optimal
strategy is defined as in Definition \ref{Definition 1}; however,
all subsequent analyses in that paper proceed under the assumption
that integrability is satisfied, leading to its main theorems being
stated for finite expectations. In fact, the technique we discuss
next is also utilized in that paper, but only to investigate the continuity
of the maximal expected logarithmic return of a portfolio and the
existence of the maximizer. In contrast, we apply the same technique
but use it to bypass the required integrability by bounding the variables,
thus expanding the market scenarios in which expectations are not
well-defined.
\end{rem*}
\textbf{Variable transformation}. In order to guarantee the well-definedness
of the conditional expected difference between the logarithmic returns
of portfolios, we show that any individual conditional expectation
can be decomposed into the sum of a well-defined expected value and
a common redundant part, which is then ignorable by subtraction. Specifically,
by adapting the following approach in a manner similar to that originally
used by \citet{Algoet1988} to demonstrate the existence of log-optimal
portfolios, any expected logarithmic portfolio return can be decomposed
by performing a transformation of the probability measure through
a change of variables.\smallskip

Let us consider the equally allocated portfolio $\hat{b}\coloneqq\big(1/m,...,1/m\big)$
as a reference portfolio and the variable of assets' returns $X$.
We subsequently define the variable $U$ of \emph{normalized assets'
returns} as the following scaling function of $X$ with parameter
$\hat{b}$:
\begin{equation}
U\coloneqq u\left(X\right)\coloneqq X/\big<\hat{b},X\big>\in\mathcal{U},\label{normalized U 1}
\end{equation}
where $\mathcal{U}$ denotes the corresponding set of all normalized
assets' returns as:
\begin{equation}
\mathcal{U}\coloneqq\left\{ u\coloneqq\big(u_{1},...,u_{m}\big)\in\mathbb{R}_{++}^{m}:\,\big<\hat{b},u\big>=1\right\} .\label{normalized U 2}
\end{equation}
Thus, for any distribution $P$ of the random variable $X$, the corresponding
distribution $Q^{P}$ of the normalized random variable $U$ is an
image measure of $P$, defined as follows:
\[
Q^{P}\big(U\in A\big)=P\big(X:\,u\big(X\big)\in A\big),\,\forall A\subseteq\mathcal{U}.
\]
As a result, while the support of the original distribution $P$ ranges
over the whole space $\mathbb{R}_{+}^{m}$, the support of the transformed
distribution $Q^{P}$ is constrained to a bounded region. This fact,
along with other critical properties relevant to our analysis, which
result directly from the transformation of the probability measure
for decomposition, are summarized in Lemma \ref{Lemma 1}.
\begin{lem}
\textup{Consider the reference portfolio $\hat{b}$ and the variable
of normalized assets' returns $U$, defined in terms of the variable
$X$ according to (\ref{normalized U 1}) and (\ref{normalized U 2}).
We have: \label{Lemma 1}}
\begin{enumerate}
\item Given a process of assets' returns $\big\{ X_{n}\big\}_{n=1}^{\infty}$,
the difference in growth rates between two generic strategies $\big(\bar{b}_{n}\big)$
and $\big(\tilde{b}_{n}\big)$ is always identical to the corresponding
difference associated with the normalized assets' returns $U_{n}=u\left(X_{n}\right)$
over time, i.e.,
\[
W_{n}\big(\bar{b}_{n}\big)-W_{n}\big(\tilde{b}_{n}\big)=\dfrac{1}{n}\sum_{i=1}^{n}\log\big<\bar{b}_{n},U_{n}\big>-\dfrac{1}{n}\sum_{i=1}^{n}\log\big<\tilde{b}_{n},U_{n}\big>,\text{ }\forall n.
\]
\item For any distribution, the corresponding expected difference between
the logarithmic returns of two portfolios remains unchanged through
the change of variable, as:
\[
\mathbb{E}\big(\log\big<\bar{b},X\big>-\log\big<\tilde{b},X\big>\big)=\mathbb{E}\big(\log\big<\bar{b},U\big>-\log\big<\tilde{b},U\big>\big),\text{ }\forall\bar{b},\tilde{b}\in\mathcal{B}^{m}.
\]
Additionally, \textup{$0\leq\max_{b\in\mathcal{B}^{m}}\mathbb{E}\big(\log\big<b,U\big>\big)\leq\log\big(m\big)$.}
\end{enumerate}
\end{lem}
\begin{proof}
Since the scaling function is defined in (\ref{normalized U 1}) as
$u\big(X\big)=X/\big<\hat{b},X\big>$, given the variable $X$, the
logarithmic returns of any portfolio can be decomposed into the following
sum:
\begin{equation}
\log\big<b,X\big>=\log\big<b,u\big(X\big)\big>+\log\big<\hat{b},X\big>,\text{ }\forall b\in\mathcal{B}^{m}.\label{decomposition 1}
\end{equation}
Therefore, given any two generic strategies $\big(\bar{b}_{n}\big)$
and $\big(\tilde{b}_{n}\big)$ and the process $\big\{ X_{n}\big\}_{n=1}^{\infty}$,
we have:
\begin{align*}
W_{n}\big(\bar{b}_{n}\big)-W_{n}\big(\tilde{b}_{n}\big) & =\dfrac{1}{n}\sum_{i=1}^{n}\log\big<\bar{b}_{n},X_{n}\big>-\dfrac{1}{n}\sum_{i=1}^{n}\log\big<\tilde{b}_{n},X_{n}\big>\\
 & =\dfrac{1}{n}\sum_{i=1}^{n}\log\big<\bar{b}_{n},u\left(X_{n}\right)\big>-\dfrac{1}{n}\sum_{i=1}^{n}\log\big<\tilde{b}_{n},u\left(X_{n}\right)\big>,\text{ }\forall n,
\end{align*}
which establishes assertion (a), given that $U_{n}=u\big(X_{n}\big)$.\smallskip

To prove assertion (b), we first note that the variable $U$ is bounded
as:
\begin{equation}
0<U_{i}=\frac{X_{i}}{\big<\hat{b},X\big>}=\frac{mX_{i}}{\sum_{i=1}^{m}X_{i}}\leq m,\text{ }\forall i\in\left\{ 1,...,m\right\} ,\label{U is bounded}
\end{equation}
due to $\hat{b}=\big(1/m,...,1/m\big)$. As a result, for any portfolio,
we have the following inequality:
\[
0=\mathbb{E}\big(\log\big<\hat{b},U\big>\big)\leq\max_{b\in\mathcal{B}^{m}}\mathbb{E}\big(\log\big<b,U\big>\big)\leq\mathbb{E}\big(\log\big(m\big)\big)=\log\big(m\big),
\]
which attests to the finite non-negativity of the maximal expected
value $\max_{b\in\mathcal{B}^{m}}\mathbb{E}\big(\log\big<b,U\big>\big)$.
Furthermore, we also obtain the decomposition deduced directly from
(\ref{decomposition 1}) as follows:
\[
\mathbb{E}\big(\log\big<b,X\big>\big)=\mathbb{E}\big(\log\big<b,U\big>\big)+\mathbb{E}\big(\log\big<\hat{b},X\big>\big).
\]
Then, the needed equality of the assertion follows immediately as:
\begin{align*}
\mathbb{E}\big(\log\big<\bar{b},X\big>-\log\big<\tilde{b},X\big>\big) & =\mathbb{E}\big(\log\big<\bar{b},U\big>\big)-\mathbb{E}\big(\log\big<\tilde{b},U\big>\big)\\
 & =\mathbb{E}\big(\log\big<\bar{b},U\big>-\log\big<\tilde{b},U\big>\big),\text{ }\forall\bar{b},\tilde{b}\in\mathcal{B}^{m},
\end{align*}
which finalizes our proof.
\end{proof}
The properties established in Lemma \ref{Lemma 1} are helpful for
our discussion in the sense that the log-optimal portfolios defined
according to Definition \ref{Definition 1}, in which the random variables
$X_{n}$ and thus the associated logarithmic returns of portfolios
$\log\big<b_{n}^{*},X_{n}\big>$ might not be integrable, can be investigated
through the equivalent conditional expected difference as:
\begin{equation}
\mathbb{E}\big(\log\big<b,X_{n}\big>-\log\big<b_{n}^{*},X_{n}\big>|\mathcal{F}_{n}^{X}\big)=\mathbb{E}\big(\log\big<b,U_{n}\big>-\log\big<b_{n}^{*},U_{n}\big>|\mathcal{F}_{n}^{X}\big)\leq0,\,\forall b\in\mathcal{B}^{m},\forall n,\label{log-optimal normalized}
\end{equation}
in which the optimally conditional expected value $\max_{b\in\mathcal{B}^{m}}\mathbb{E}\big(\log\big<b,U\big>|\mathcal{F}_{n}^{X}\big)$
is always well-defined, i.e., the variables $\log\big<b_{n}^{*},U\big>$
associated with the log-optimal portfolios $b_{n}^{*}$ are integrable
given any probability measure. Moreover, the approach of considering
the difference between the logarithmic returns of two portfolios rather
than their individual values allows us to remove the redundantly inherent
term $\log\big<\hat{b},X\big>$, thus reducing the originally unbounded
growth rate of the log-optimal strategy $W_{n}\big(b_{n}^{*}\big)$
to the upper-bounded growth rate $n^{-1}\sum_{i=1}^{n}\log\big<b_{n}^{*},U_{n}\big>$
over all periods. In this context, some important properties of the
log-optimal strategy under the normalized assets' returns will be
discussed in the next section.

\section{Supporting lemmas on the utilization of market information}

After introducing the log-optimal strategy and the normalized assets'
returns variable, we investigate various properties of competitive
strategies benchmarked against the log-optimal strategy in this section.
Through these comparisons, the optimality of the benchmarked log-optimal
strategy is demonstrated in terms of utilizing market information
to asymptotically achieve superior growth rates and expected cumulative
wealth. These properties also provide further insights into the maximum
potential of a strategy when all side information is fully observed.
The analysis is presented in the following three subsections, where
some results are straightforward extensions of those by \citet{Algoet1988},
adapted to the normalized assets' returns variable, while others are
improvements, offering clearer proofs and novel deductions.

\subsection{Optimality principles of the log-optimal strategy}

In addition to being optimal in maximizing the conditional expected
value of the logarithmic return of a portfolio, log-optimal portfolios
are also optimal in guaranteeing the highest conditional expected
ratio of return compared to all other portfolios, given any $\sigma$-field
$\mathcal{F}_{n}^{X}$. Consequently, the growth rate of the log-optimal
strategy will almost surely (a.s.) not be exceeded by that of any
other competitive strategy infinitely often by an arbitrary magnitude
as time evolves. These optimal properties of the log-optimal strategy
are established in Lemma \ref{Lemma 2}.
\begin{lem}
Consider a log-optimal strategy $\big(b_{n}^{*}\big)$ and any competitive
strategy $\big(b_{n}\big)$. We have: \label{Lemma 2}
\begin{enumerate}
\item $\mathbb{E}\big(\big<b_{n},U_{n}\big>\big/\big<b_{n}^{*},U_{n}\big>|\mathcal{F}_{n}^{X}\big)\leq1,\text{ }\forall n.$
\item $\limsup_{n\to\infty}\big(W_{n}\big(b_{n}\big)-W_{n}\big(b_{n}^{*}\big)\big)\leq0\text{, a.s}.$
\end{enumerate}
Noting that the growth rate of the log-optimal strategy $W_{n}\left(b_{n}^{*}\right)$
does not necessarily have a limit.
\end{lem}
\begin{proof}
Assertion (a) is simply a direct result from applying the Kuhn-Tucker
conditions for log-optimality in \citet{Bell1980} and \citet{Algoet1988}.
Specifically, noting that the log-optimal portfolios can be equivalently
obtained as (\ref{log-optimal normalized}) by Lemma \ref{Lemma 1},
where the maximal conditional expected values $\max_{b\in\mathcal{B}^{m}}\mathbb{E}\big(\log\big<b,U\big>|\mathcal{F}_{n}^{X}\big)$
are finite, the log-optimal portfolios always exist, and the Kuhn-Tucker
conditions yield the required inequality in the assertion.\smallskip

To prove assertion (b), we modify the arguments originally presented by \citet{Algoet1988}
to show that $\big(S_{n}\big(b_{n}\big)/S_{n}\big(b_{n}^{*}\big),\mathcal{F}_{n}^{X}\big)$
is a non-negative supermartingale. In detail, since $\big<\hat{b},X_{n}\big>>0$
for any $n$, we obtain the following inequality using assertion (a):
\begin{align*}
\mathbb{E}\Big(\frac{S_{n+1}\big(b_{n+1}\big)}{S_{n+1}\big(b_{n+1}^{*}\big)}\Big|\mathcal{F}_{n+1}^{X}\Big) & =\mathbb{E}\Big(\frac{S_{n}\big(b_{n}\big)\big<\hat{b},X_{n+1}\big>\big<b_{n+1},U_{n+1}\big>}{S_{n}\big(b_{n}^{*}\big)\big<\hat{b},X_{n+1}\big>\big<b_{n+1}^{*},U_{n+1}\big>}\Big|\mathcal{F}_{n+1}^{X}\Big)\\
 & =\frac{S_{n}\big(b_{n}\big)}{S_{n}\big(b_{n}^{*}\big)}\mathbb{E}\Big(\frac{\big<b_{n+1},U_{n+1}\big>}{\big<b_{n+1}^{*},U_{n+1}\big>}\Big|\mathcal{F}_{n+1}^{X}\Big)\leq\frac{S_{n}\big(b_{n}\big)}{S_{n}\big(b_{n}^{*}\big)},\forall n.
\end{align*}
Consequently, we obtain $\mathbb{E}\big(S_{n}\big(b_{n}\big)/S_{n}\big(b_{n}^{*}\big)\big)\leq\mathbb{E}\big(S_{0}/S_{0}\big)=1$
and then deduce the required result of the assertion by invoking the
Markov inequality as follows:
\begin{align*}
\mathbb{\mathbb{P}}\big(S_{n}\big(b_{n}\big)>n^{2}S_{n}\big(b_{n}^{*}\big) & <\frac{1}{n^{2}}\\
\Rightarrow\,\,\,\,\,\,\,\,\,\,\,\,\,\,\,\,\,\,\,\,\,\,\,\sum_{n=1}^{\infty}\mathbb{\mathbb{P}}\Big(\frac{1}{n}\log\frac{S_{n}\left(b_{n}\right)}{S_{n}\left(b_{n}^{*}\right)}>\frac{2\log n}{n}\Big) & \leq\sum_{n=1}^{\infty}\frac{1}{n^{2}}=\frac{\pi^{2}}{6}\,\text{(the Basel problem)}\\
\Rightarrow\mathbb{\mathbb{P}}\Big(\frac{1}{n}\log\frac{S_{n}\left(b_{n}\right)}{S_{n}\left(b_{n}^{*}\right)}>\epsilon\text{ infinitely often}\Big) & =0\,\text{ (by the Borel-Cantelli lemma),}\text{ }
\end{align*}
for any arbitrarily small $\epsilon>0$ as $n\to\infty$, i.e $\limsup_{n\to\infty}\big(W_{n}\big(b_{n}\big)-W_{n}\big(b_{n}^{*}\big)\big)\leq0$
almost surely.
\end{proof}
Besides the almost sure guarantee of optimality in asymptotic growth
rate compared with other competitive strategies by Lemma \ref{Lemma 2},
the log-optimal strategy is also the optimal one that maximizes the
expected logarithmic cumulative wealth among all accessible strategies
for any time horizon. This property is discussed in the next section,
which evaluates the benefit of using a larger amount of information
rather than a smaller amount when making decisions.

\subsection{Positive impact of large information size}

One might argue that the use of all information, including unobservable
side information, is very costly and challenging due to the difficulty
of accurately estimating the conditional distribution, particularly
since the side information is often high-dimensional in reality. However,
if there is a relationship between the assets' returns and all the
side information, removing some of this information to reduce dimensionality
only makes sense if its correlation with the assets' returns is insignificant
- something that, in fact, remains unknown if not all information
can be fully observed. In Lemma \ref{Lemma 3}, we establish a comparison
between two log-optimal portfolios that depend on different sizes
of observed information. This comparison highlights the advantages
of utilizing a larger set of information and justifies the optimality
of the log-optimal strategies over any other kind of strategy among
all possible ones.
\begin{lem}
Assume an assets' returns variable $X$, defined on $\big(\Omega,\mathbb{F},\mathbb{\mathbb{P}}\big)$,
admits a regular conditional distribution given any sub-$\sigma$-field
of $\mathbb{F}$. Then, the normalized assets' returns variable $U$
also admits the corresponding transformed conditional distribution,
and the following hold:\label{Lemma 3}
\begin{enumerate}
\item Given a sub-$\sigma$-field $\mathcal{G}\subseteq\mathbb{F}$, there
exists a $\mathcal{G}$-measurable optimal portfolio such that:
\[
\max_{b\in\mathcal{B}^{m}}\mathbb{E}\big(\log\big<b,U\big>|\mathcal{G}\big)=\max_{b\in\mathcal{G}}\mathbb{E}\big(\log\big<b,U\big>|\mathcal{G}\big),
\]
where the notation $b\in\mathcal{G}$ denotes the $\mathcal{G}$-measurability
of a portfolio $b$. Consequently:
\[
\mathbb{E}\big(\max_{b\in\mathcal{B}^{m}}\mathbb{E}\big(\log\big<b,U\big>|\mathcal{G}\big)\big)=\max_{b\in\mathcal{G}}\mathbb{E}\big(\log\big<b,U\big>\big).
\]
\item The following inequality holds for any two sub-$\sigma$-fields $\mathcal{G}\subseteq\mathcal{H}\subseteq\mathbb{F}$:
\[
\mathbb{E}\big(\max_{b\in\mathcal{B}^{m}}\mathbb{E}\big(\log\big<b,U\big>|\mathcal{G}\big)\big)\leq\mathbb{E}\big(\max_{b\in\mathcal{B}^{m}}\mathbb{E}\big(\log\big<b,U\big>|\mathcal{H}\big)\big).
\]
\end{enumerate}
\end{lem}
\begin{proof}
Assertion (a) follows directly from Theorem 2 in \citet{Algoet1988},
which guarantees the existence of a measurably selected optimal portfolio,
depending on $Q^{P}$, that maximizes the expectation $\mathbb{E}\big(\log\big<b,U\big>\big)$
among all portfolios in $\mathcal{B}^{m}$. Here, $U\sim Q^{P}$,
and $Q^{P}$ is the transformation of the distribution $P$ of the
variable $X$. Since the conditional distribution given the $\sigma$-field
$\mathcal{G}$ is $\mathcal{G}$-measurable, there exists an optimal
portfolio, denoted as $b_{\mathcal{G}}^{*}\in\mathcal{G}$, as a maximizer
for $\max_{b\in\mathcal{B}^{m}}\mathbb{E}\big(\log\big<b,U\big>|\mathcal{G}\big)$.
Then, taking the expectation on both sides of the equality $\max_{b\in\mathcal{B}^{m}}\mathbb{E}\big(\log\big<b,U\big>|\mathcal{G}\big)=\max_{b\in\mathcal{G}}\mathbb{E}\big(\log\big<b,U\big>|\mathcal{G}\big)$,
we obtain the following:
\begin{align*}
\mathbb{E}\big(\max_{b\in\mathcal{B}^{m}}\mathbb{E}\big(\log\big<b,U\big>|\mathcal{G}\big)\big) & =\mathbb{E}\big(\mathbb{E}\big(\log\big<b_{\mathcal{G}}^{*},U\big>|\mathcal{G}\big)\big)\\
 & =\max_{b\in\mathcal{G}}\mathbb{E}\big(\log\big<b,U\big>\big),
\end{align*}
due to the law of total expectation applied to the random variable
$\log\big<b_{\mathcal{G}}^{*},U\big>$.\smallskip

Now, consider any two sub-$\sigma$-fields $\mathcal{G}\subseteq\mathcal{H}\subseteq\mathbb{F}$.
There exist two optimal portfolios, denoted as $b_{\mathcal{G}}^{*}\in\mathcal{G}$
and $b_{\mathcal{H}}^{*}\in\mathcal{H}$, for $\max_{b\in\mathcal{B}^{m}}\mathbb{E}\big(\log\big<b,U\big>|\mathcal{G}\big)$
and $\max_{b\in\mathcal{B}^{m}}\mathbb{E}\big(\log\big<b,U\big>|\mathcal{H}\big)$,
respectively. In addition, since $\mathcal{G}\subseteq\mathcal{H}$
implies $b_{\mathcal{G}}^{*}\in\mathcal{H}$, we have:
\begin{align*}
\mathbb{E}\big(\max_{b\in\mathcal{B}^{m}}\mathbb{E}\big(\log\big<b,U\big>|\mathcal{G}\big)\big) & =\mathbb{E}\big(\log\big<b_{\mathcal{G}}^{*},U\big>\big)\\
 & \leq\max_{b\in\mathcal{H}}\mathbb{E}\big(\log\big<b,U\big>\big)=\mathbb{E}\big(\max_{b\in\mathcal{B}^{m}}\mathbb{E}\big(\log\big<b,U\big>|\mathcal{H}\big)\big).
\end{align*}
This demonstrates assertion (b) and completes the proof.
\end{proof}
It is worth stressing again here that we are considering all possible
strategies $\big(b_{n}\big)$, where each causal portfolio $b_{n}$
in the sequence is selected based solely on past market information
up to time $n$. These strategies obviously include all the log-optimal
strategies. Hence, for any time horizon $m$, a log-optimal strategy
maximizes the expected logarithm of cumulative wealth among all strategies
in terms of normalized asset returns, as shown by the following chain
rule:
\begin{align}
\max_{\left(b_{n}\right)}\mathbb{E}\big(\sum_{i=1}^{m}\log\big<b_{i},U_{i}\big>\big) & =\sum_{i=1}^{m}\max_{b_{i}}\mathbb{E}\big(\log\big<b_{i},U_{i}\big>\big)\nonumber \\
 & =\sum_{i=1}^{m}\mathbb{E}\big(\max_{b\in\mathcal{B}^{m}}\mathbb{E}\big(\log\big<b,U_{i}\big>|\mathcal{F}_{i}^{X}\big)\big),\text{ }\forall m.\label{chain rule}
\end{align}
It should be emphasized that the optimality property described above
is general, in that it does not require specific functional characteristics
or measurability for the log-optimal strategy or any other strategy.
Consequently, any other strategy that causally uses only a smaller
subset of past information as a sub-$\sigma$-field of $\mathcal{F}_{n}^{X}$
at all times $n$, might be inferior in terms of the expected logarithm
of cumulative wealth due to $\mathbb{E}\big(\max_{b\in\mathcal{B}^{m}}\mathbb{E}\big(\log\big<b,U_{n}\big>|\mathcal{G}\big)\big)\leq\mathbb{E}\big(\max_{b\in\mathcal{B}^{m}}\mathbb{E}\big(\log\big<b,U_{n}\big>|\mathcal{F}_{n}^{X}\big)\big)$
for any $\sigma$-field $\mathcal{G}\subseteq\mathcal{F}_{n}^{X}$,
as shown in Lemma \ref{Lemma 3}. Moreover, the corresponding expected
growth rate of a strategy, in terms of normalized asset returns, is
also maximized by the log-optimal strategies for any time horizon.
This expected growth rate will almost surely converge to a finite
limit as time progresses, provided the market is stationary, as demonstrated
later in Corollary \ref{Corollary 1}.

\subsection{Limit behavior of the log-optimal strategy}

Since the log-optimal portfolio is derived using the conditional distributions
based on past observations over time, it is essential to investigate
its limit behavior in terms of the distributions, as described in
Lemma \ref{Lemma 4}. Generally, the expected logarithmic returns
of the portfolio, in terms of normalized asset returns, are continuous
in the space of all distributions equipped with the weak topology.
This implies that, given a sequence of converging distributions, any
accumulation of the corresponding log-optimal portfolios must also
be a log-optimal portfolio for the limiting distribution. Hence, as
the regular conditional distribution, given an increasing information
field, is a converging sequence, any choice of log-optimal strategy
ensures that the portfolio's returns will almost surely converge to
those of the limiting portfolio under the limiting distribution.\smallskip

Moreover, another important consequence of the lemma is that the normalized
return of the log-optimal portfolios $\big<b_{n}^{*},U_{n}\big>$
is uniquely defined for almost all outcomes of $U_{n}$ under the
associated conditional distributions, regardless of the choice of
$b_{n}^{*}$. The log-optimal $b_{n}^{*}$ is uniquely defined only
if the associated regular conditional distribution $\mathbb{P}\big(U_{n}|\mathcal{F}_{n}^{X}\big)$
has full-dimensional support, which makes the corresponding expected
values $\mathbb{E}\big(\log\big<b,U_{n}\big>|\mathcal{F}_{n}^{X}\big)$
strictly concave in $b\in\mathcal{B}^{m}$. Consequently, any log-optimal
strategy $\big(b_{n}^{*}\big)$ must also converge to a limiting portfolio
almost surely if the conditional distribution $\mathbb{P}\big(U_{1}|\mathcal{F}_{\infty}^{X}\big)$
almost surely has full-dimensional support.
\begin{lem}
With $P$-a.a $X$ representing the extension notation for an event
occurring for almost all outcomes of $X$ according to a specific
distribution $P$ of $X$, we have:\label{Lemma 4}
\begin{enumerate}
\item Consider two portfolios $\bar{b}$ and $\tilde{b}$ maximizing $\mathbb{E}\big(\log\big<b,U\big>\big)$,
where $U\sim Q$. Then:
\[
\big<\bar{b},U\big>=\big<\tilde{b},U\big>\text{, \ensuremath{Q}-a.a }U.
\]
\item Let $b^{*}\big(Q\big)$ denote a portfolio maximizing the expected
portfolio's logarithmic return in terms of normalized assets' returns
according to the distribution $Q$ as $\mathbb{E}\big(\log\big<b^{*}\big(Q\big),U\big>\big)=\max_{b\in\mathcal{B}^{m}}\mathbb{E}\big(\log\big<b,U\big>\big)$
where $U\sim Q$. Consider any sequence of distributions $\big\{ Q_{n}\big\}_{n=1}^{\infty}$
converging weakly to a limiting distribution $Q_{\infty}$ as $n\to\infty$,
then:
\[
\begin{cases}
\lim_{n\to\infty}\big<b^{*}\big(Q_{n}\big),U\big>=\big<b^{*}\big(Q_{\infty}\big),U\big>\text{, \ensuremath{Q_{\infty}}-a.a }U\\
\lim_{n\to\infty}\mathbb{E}\big(\log\big<b^{*}\big(Q_{n}\big),U\big>\big)=\mathbb{E}\big(\log\big<b^{*}\big(Q_{\infty}\big),U\big>\big)
\end{cases}
\]
\item Consider an assets' returns variable $X$ defined on $\big(\Omega,\mathbb{F},\mathbb{\mathbb{P}}\big)$
and admitting a regular conditional distribution given the $\sigma$-fields
$\mathcal{G}_{n}\subseteq\mathbb{F}$ for all $n$ and $\mathcal{G}_{\infty}\coloneqq\sigma\big(\cup_{n=1}^{\infty}\mathcal{G}_{n}\big)\subseteq\mathbb{F}$,
then the following equality holds for the normalized assets' returns
$U$ of $X$:
\[
\lim_{n\to\infty}\mathbb{P}\big(X|\mathcal{G}_{n}\big)=\mathbb{P}\big(X|\mathcal{G}_{\infty}\big)\text{, weakly a.s, implying }\lim_{n\to\infty}\big<\bar{b}_{n},U\big>=\big<\bar{b}_{\infty},U\big>,\text{ a.s},
\]
where the portfolios $\bar{b}_{n}$ are maximizers for $\max_{b\in\mathcal{B}^{m}}\mathbb{E}\big(\log\big<b,U\big>|\mathcal{G}_{n}\big)$
with $n\in\mathbb{N}_{+}\cup\left\{ \infty\right\} $.
\end{enumerate}
\end{lem}
\begin{proof}
To prove assertion (a), we utilize the Kuhn-Tucker condition for the
log-optimal portfolio. Specifically, since $\mathbb{E}\big(\big<\bar{b},U\big>\big/\big<\tilde{b},U\big>\big)=1$
where $U\sim Q$ by Lemma \ref{Lemma 2}, we have:
\[
0=\mathbb{E}\Big(\log\frac{\big<\bar{b},U\big>}{\big<\tilde{b},U\big>}\Big)\leq\log\mathbb{E}\Big(\frac{\big<\bar{b},U\big>}{\big<\tilde{b},U\big>}\Big)=0,
\]
due to Jensen's inequality, where the equality holds if and only if
$\big<\bar{b},U\big>/\big<\tilde{b},U\big>=1$ for $Q$-almost all
$U$, as required for the assertion.\smallskip

Regarding assertion (b), we note that the bottom equality in the statement
is directly justified by the continuity of the function $\max_{b\in\mathcal{B}^{m}}\mathbb{E}\big(\log\big<b,U\big>\big)$
in the distribution $Q$ of $U$, as demonstrated Theorem 2 of \citet{Algoet1988}.
Meanwhile, the upper equality in the statement is only a comment on
Theorem 3 in this paper, but without a clear argument. Thus, we provide
a clarified deduction. Specifically, suppose there exists a unique
maximizer $b^{*}\big(Q_{\infty}\big)$ for the problem $\max_{b\in\mathcal{B}^{m}}\mathbb{E}\big(\log\big<b,U\big>\big)$
where $U\sim Q_{\infty}$, then any sequence $\left\{ b^{*}\big(Q_{n}\big)\right\} _{n=1}^{\infty}$
must converge to $b^{*}\big(Q_{\infty}\big)$, and so $\left\{ \big<b^{*}\big(Q_{n}\big),U\big>\right\} _{n=1}^{\infty}$
must converge to $\big<b^{*}\big(Q_{\infty}\big),U\big>$ for almost
all outcomes of $U$. Meanwhile, in the case where there are many
maximizers for $\max_{b\in\mathcal{B}^{m}}\mathbb{E}\big(\log\big<b,U\big>\big)$,
the sequence $\left\{ b^{*}\big(Q_{n}\big)\right\} _{n=1}^{\infty}$
may not converge, but all its accumulation points must be among these
maximizers. Furthermore, since $\big<b^{*}\big(Q_{\infty}\big),U\big>$
must be unique for any choice of $b^{*}\big(Q_{\infty}\big)$, as
guaranteed by assertion (a) for $Q$-almost all $U$, the sequence
$\left\{ \big<b^{*}\big(Q_{n}\big),U\big>\right\} _{n=1}^{\infty}$
must converge to $\big<b^{*}\big(Q_{\infty}\big),U\big>$ for $Q$-almost
all $U$, as needed for the statement.\smallskip

The proof for assertion (c) follows from Theorem 4 in the aforementioned
paper, but with a clarified deduction. Specifically, since $\mathbb{P}\big(X|\mathcal{G}_{n}\big)\to\mathbb{P}\big(X|\mathcal{G}_{\infty}\big)$
weakly almost surely by Levy's martingale convergence theorem, we
have $\big<\bar{b}_{n},U\big>\to\big<\bar{b}_{\infty},U\big>$ for
$\mathbb{P}\big(U|\mathcal{G}_{\infty}\big)$-almost all $U$, almost
surely by assertion (b). Therefore, we conclude that $\big<\bar{b}_{n},U\big>\to\big<\bar{b}_{\infty},U\big>$
almost surely, completing the proof.
\end{proof}

\section{Limit and equality theorems of growth rate in stationary market}

Since information is gradually and continuously revealed, the amount
of past data used for decision-making, including the unobservable
events, increases to infinitely large amounts. This probably misleads
us to think that a dynamic log-optimal strategy, which varies depending
on increasing information fields and is optimal in both the sense
of asymptotic growth rate and expected cumulative wealth and growth
rate (as discussed in the previous sections), should also easily outperform
any constant strategy, which is invariant under changing market information.
In fact, as we mentioned in the introduction as our motivation, the
status quo of the dominating opinion in the relevant contemporary
literature reflects this evaluation, following the well-known theories
established in \citet{Algoet1988,Algoet1992,Cover1991,Cover1996,Cover1996a,Cover2006}.
However, the results established in Theorem \ref{Theorem 1}, Theorem
\ref{Theorem 3} and later Corollary \ref{Corollary 1} demonstrate
surprising findings that contradict this popular unsettled conjecture.\smallskip

Before investigating the evolutionary behavior of the log-optimal
strategy, let us remark on some properties of stationary processes
and introduce a new notation. Given the operator $T$ representing
the metrically transitive and invertible measure-preserving transformation
of the underlying probability space $\text{\ensuremath{\big(\Omega},\ensuremath{\mathbb{F}},\ensuremath{\mathbb{\mathbb{P}}}}\big)$,
let $\mathcal{I}\subset\mathbb{F}$ denote the \emph{invariant} sub-$\sigma$-field
henceforth as:
\[
\mathcal{I}\coloneqq\big\{ F\in\mathbb{F}:\,T^{-1}F=F\big\}.
\]
A stationary pairwise process $\big\{ X_{n},Y_{n}\big\}_{n=1}^{\infty}$
is defined as $\big(X_{n},Y_{n}\big)\big(\omega\big)=\big(X_{1},Y_{1}\big)\big(T^{(n-1)}\omega\big)$
for all $n$. This process could be extended infinitely to the left
as a two-sided process with the origin at $n=1$ by Kolmogorov's extension
theorem. Under this condition, the invariant $\sigma$-field $\mathcal{I}$
is contained in the completion of the $\sigma$-field $\mathcal{F}_{\infty}$
of infinite information (and thus also in the completion of the $\sigma$-field
$\mathcal{F}_{\infty}^{X}$), which implies that no finite amount
of information, but only the infinite one, suffices to exactly determine
the random ergodic mode of the process $\big\{ X_{n},Y_{n}\big\}_{n=1}^{\infty}$.
Fortunately, the true mode among a mixture of stationary ergodic ones
could be gradually identified by the log-optimal strategy, which depends
on the continuously increasing information fields. For convenience,
let us introduce the reduced case of a \emph{discrete market} in Definition
\ref{Def-of-discrete-market} and consider it the implicit default
market henceforth, unless otherwise specified.
\begin{defn}
A discrete market is defined as a pairwise process $\big\{ X_{n},Y_{n}\big\}_{n=1}^{\infty}$
in which the assets' returns variables $X_{n}$ take values in a countable
subset $\mathbb{X}\subset\mathbb{R}_{++}^{m}$, for all $n$. In other
words, they follow discrete (purely atomic) distributions with support
$\mathbb{X}$.\label{Def-of-discrete-market}\smallskip 
\end{defn}
\begin{rem*}
It is worth noting that the introduction of the discrete market in
Definition \ref{Def-of-discrete-market} is merely to simplify the
proof of Lemma \ref{Lemma 5}, and it does not affect the subsequent
theorems, which remain valid without this condition. However, since
the assets' returns in the discrete market may take an infinite number
of values, the loss of generality is minimal, particularly when considering
real-world markets. Additionally, the normalized assets' returns variable
$U$ is bounded within a small set as in (\ref{U is bounded}), implying
that a fine discretization can cover most, if not all, possible outcomes.
In the proof of Proposition \ref{Proposition 1}, which appears later,
an approach to bypass Lemma \ref{Lemma 5} is discussed, so the discrete
market condition is not required, though this involves a few more
auxiliary lemmas, making Theorem \ref{Theorem 1} less straightforward.
\end{rem*}

\subsection{Equilibrium state under infinitely increasing market information}

Lemma \ref{Lemma 4} provides critical tools to argue Theorem \ref{Theorem 1},
which serves as the basis to affirm that the periodic portfolios'
returns of the log-optimal strategy could not exceed those of all
constant strategies infinitely often as time progresses. This implies
that such a dynamically optimal strategy could outperform any constant
strategy, by utilizing the entirety of the increasing past information,
only over a finite time horizon; thereafter, this advantage of using
past information gradually degrades into an equilibrium state where
the difference between the portfolios' returns of the dynamically
optimal strategy and some constant strategies becomes negligible over
time. As the theorem states, the cause of this phenomenon is clearly
the stationarity of the market process, which directs the evolution
of the log-optimal portfolios toward some optimal ones with respect
to the limiting distribution, given the infinite past events as the
largest possible capacity of the information field. Roughly speaking,
although the cumulative wealth of the log-optimal strategy could be
significantly better, its growth rate will be asymptotically identical
to that of some constant strategies once it enters the equilibrium
state.\smallskip 
\begin{lem}
Assume the market process $\big\{ X_{n},Y_{n}\big\}_{n=1}^{\infty}$
is stationary. Then, the following convergence holds for any two generic
strategies $\big(\bar{b}_{n}\big)$ and $\big(\tilde{b}_{n}\big)$,
such that $\big(\log\big<\bar{b}_{n},U_{k}\big>-\log\big<\tilde{b}_{n},U_{k}\big>\big)\to0$
almost surely for all normalized assets' returns $U_{k}$ with $k\geq1$:\label{Lemma 5}
\[
\lim_{n\to\infty}\big(\log\big<\bar{b}_{n},U_{n}\big>-\log\big<\tilde{b}_{n},U_{n}\big>\big)=0,\text{ a.s}.
\]
\end{lem}
\begin{proof}
By stationarity, noting that any finite sample $\left\{ X_{1}\big(\omega\big),X_{2}\big(\omega\big),...,X_{n}\big(\omega\big)\right\} $
can be represented as $\left\{ X_{1}\big(\omega\big),X_{1}\big(T\omega\big)...,X_{1}\big(T^{(n-1)}\omega\big)\right\} $
for each $\omega\in\Omega$ and for all $n$, the ergodic mode of
the market process is determined by the invariant $\sigma$-field
$\mathcal{I}$. Hence, by applying Birkhoff's ergodic theorem, we
obtain the following limit (with the discrete market as in Definition
\ref{Def-of-discrete-market}):
\[
\lim_{n\to\infty}\frac{1}{n}\sum_{i=1}^{n}\mathbb{I}_{\left\{ X_{1}=x\right\} }=\mathbb{P}\big(X_{1}=x|\mathcal{I}\big)>0,\text{ a.s}.
\]
This implies that for each realization of the process $\big\{ X_{n}\big(\omega\big)\big\}_{n=1}^{\infty}$,
there exists a corresponding subset $\mathbb{X}^{\omega}\subset\mathbb{X}$
such that all outcomes $x\in\mathbb{X}^{\omega}$ with $\mathbb{P}\big(X_{1}=x|\mathcal{I}\big)\big(\omega\big)>0$
must occur infinitely often as $X_{n}\big(\omega\big)=x$ for $n\to\infty$,
while other $x\notin\mathbb{X}^{\omega}$ occur only finitely many
times. Therefore, the convergence $\big(\log\big<\bar{b}_{n}\big(\omega\big),u\big(x\big)\big>-\log\big<\tilde{b}_{n}\big(\omega\big),u\big(x\big)\big>\big)\to0$
holds for $\mathbb{P}\big(X_{1}|\mathcal{I}\big)\big(\omega\big)$-almost
all $x\in\mathbb{X}^{\omega}$ almost surely, due to the assumption
that $\big(\log\big<\bar{b}_{n},U_{k}\big>\big(\omega\big)-\log\big<\tilde{b}_{n},U_{k}\big>\big(\omega\big)\big)\to0$
almost surely for all $k\geq1$. This convergence implies that for
any $\epsilon>0$ and each $x\in\mathbb{X}^{\omega}$, we have:
\[
\exists N^{\omega,\epsilon}\big(x\big)\text{ such that }|\log\big<\bar{b}_{n}\big(\omega\big),u\big(x\big)\big>-\log\big<\tilde{b}_{n}\big(\omega\big),u\big(x\big)\big>|<\epsilon,\,\forall n\geq N^{\omega,\epsilon}\big(x\big).
\]

Then, the required convergence $\big(\log\big<\bar{b}_{n},U_{n}\big>\big(\omega\big)-\log\big<\tilde{b}_{n},U_{n}\big>\big(\omega\big)\big)\to0$
must hold. If not, then there exists $\epsilon^{\omega}>0$ such that
$|\log\big<\bar{b}_{t}\big(\omega\big),u\big(x\big)\big>-\log\big<\tilde{b}_{t}\big(\omega\big),u\big(x\big)\big>|\geq\epsilon^{\omega}$
for $t\geq K$ for some $x$ as $K\to\infty$, i.e., some $x\in\mathbb{X}^{\omega}$
occurs only finitely many times because some $N^{\omega,\epsilon^{\omega}}\big(x\big)$
are never visited, which contradicts the derived property above. Hence,
we conclude that, with probability one, the convergence $\big(\log\big<\bar{b}_{n},U_{n}\big>-\log\big<\tilde{b}_{n},U_{n}\big>\big)\to0$
holds, completing the proof.
\end{proof}
\begin{thm}
If the market process $\big\{ X_{n},Y_{n}\big\}_{n=1}^{\infty}$ is
stationary, there exists a random constant strategy, denoted as $\big(b^{\omega}\big)$,
such that the growth rate of any log-optimal strategy $\big(b_{n}^{*}\big)$
does not exceed that of the constant one infinitely often by any fixed
magnitude, as:\label{Theorem 1}
\[
\lim_{n\to\infty}\big(W_{n}\big(b^{\omega}\big)-W_{n}\big(b_{n}^{*}\big)\big)=0,\text{ a.s}.
\]
Moreover, such a random constant strategy $\big(b^{\omega}\big)$
can be chosen measurably.
\end{thm}
\begin{proof}
Since the random process $\big\{ X_{n},Y_{n}\big\}_{n=1}^{\infty}$
is stationary, it can be extended to a two-sided process. Thus, any
variable $\big(X_{n},Y_{n}\big)$, for all $n$, admits a regular
conditional distribution given any $\sigma$-field $\mathcal{F}_{n}$
or the infinite $\sigma$-field $\mathcal{F}_{\infty}$. This immediately
results in $\mathbb{E}\big(\log\big<b,X_{n}\big>|\mathcal{F}_{k}^{X}\big)=\mathbb{E}\big(\log\big<b,X_{1}\big>|\mathcal{F}_{k}^{X}\big)$
for any $n>1$ and $b\in\mathcal{B}^{m}$, given the same information
field $\mathcal{F}_{k}^{X}$ with $k\in\mathbb{N}_{+}\cup\left\{ \infty\right\} $.
Then, by Lemma \ref{Lemma 4}, it follows that $\mathbb{P}\big(X_{k}|\mathcal{F}_{n}^{X}\big)\to\mathbb{P}\big(X_{k}|\mathcal{F}_{\infty}^{X}\big)$
weakly almost surely as $n\to\infty$ for all $k\geq1$, and the limit
below holds for any log-optimal strategy $\big(b_{n}^{*}\big)$:
\[
\lim_{n\to\infty}\log\big<b_{n}^{*},U_{k}\big>=\log\big<b_{\infty}^{*},U_{k}\big>,\,\forall k\geq1\text{, a.s},
\]
with portfolio $b_{\infty}^{*}$ denoting an arbitrary maximizer for
$\max_{b\in\mathcal{B}^{m}}\mathbb{E}\big(\log\big<b,U_{1}\big>|\mathcal{F}_{\infty}^{X}\big)$,
and noting that $b_{\infty}^{*}$ is random. This convergence can
be represented as $\big(\log\big<b_{n}^{*},U_{k}\big>-\log\big<b_{\infty}^{*},U_{k}\big>\big)\to0$
for the log-optimal strategy $\big(b_{n}^{*}\big)$ and the random
constant strategy $\big(b_{\infty}^{*}\big)$, for all $k\geq1$.
Hence, we have almost surely that $\big(\log\big<b_{n}^{*},U_{n}\big>-\log\big<b_{\infty}^{*},U_{n}\big>\big)\to0$
by Lemma \ref{Lemma 5}, and the following is obtained:
\begin{align*}
\lim_{n\to\infty}\big(W_{n}\big(b_{\infty}^{*}\big)-W_{n}\big(b_{n}^{*}\big)\big) & =\lim_{n\to\infty}\frac{1}{n}\sum_{i=1}^{n}\big(\log\big<b_{\infty}^{*},U_{i}\big>-\log\big<b_{i}^{*},U_{i}\big>\big)\\
 & =\lim_{n\to\infty}\big(\log\big<b_{\infty}^{*},U_{n}\big>-\log\big<b_{n}^{*},U_{n}\big>\big)=0,\text{ a.s},
\end{align*}
due to Cesaro's mean theorem and Lemma \ref{Lemma 1}. As a consequence,
by simply setting a portfolio $b^{\omega}=b_{\infty}^{*}$, which
is random and can be measurably selected depending on the $\sigma$-field
$\mathcal{F}_{\infty}^{X}$ by Lemma \ref{Lemma 3}, we finish the
proof of the existence of the optimal random constant strategy $\big(b^{\omega}\big)$.
\end{proof}

\subsection{Optimal constant strategy unconditioned on side information}

It is known from Theorem \ref{Theorem 1} that there exist some random
constant strategies, derived depending on infinitely many data of
assets' returns and side information, that could asymptotically approach
the growth rate and portfolio returns of the log-optimal strategy
under a stationary market. However, one might argue a specific context
in which there exists a non-random, i.e., independent from the market\textquoteright s
possibilities, optimal constant strategy with such a property. In
addition, an interesting question naturally arises from the existence
of the random optimal constant strategy: whether the limiting optimal
growth rate could be simultaneously represented for both the log-optimal
strategy, which is dynamic, and the constant strategies. Theorem \ref{Theorem 2}
establishes such an equality for the limiting growth rate in terms
of normalized assets' returns for the two regarded strategies under
the stationary condition. Through this equality, the optimal strategy's
growth rate when infinitely increasing market side information is
used is equivalently measured by the maximal capacity of a constant
strategy without depending on periodic market observation. Moreover,
if the market is also ergodic, the existence of a non-random optimal
constant strategy could be secured using the mentioned equality. Before
the theorem statement, let us further introduce an optional condition
for the market in Definition \ref{Def of no-trash stocks}.\smallskip 
\begin{defn}
A market is said not to have \emph{trash assets} if $\max_{1\leq i\leq m}X_{n,i}/\min_{1\leq i\leq m}X_{n,i}\geq c\gg0$
for all $n$ almost surely. This does not necessarily bound the returns
of the assets, but instead excludes the trash assets from the considered
portfolios, which are defined as those that might potentially collapse
in value to nearly zero prices due to the bankruptcy of the firms.\label{Def of no-trash stocks}\smallskip 
\end{defn}
\begin{rem*}
It is important to highlight that the almost sure convergence of the
log-optimal strategy in Theorem \ref{Theorem 2} can be proven without
the assumption of no-trash assets, by simply using Birkhoff's ergodic
theorem, as in the original proof by \citet{Algoet1988}. The key
difference is that Theorem \ref{Theorem 2} does not assume finite
expected values, due to the transformation from assets' returns to
the normalized ones. The additional no-trash assets condition is primarily
used to ensure uniform integrability, which can be useful for certain
later applications, such as more conveniently establishing Corollary
\ref{Corollary 1} and Theorem \ref{Theorem 4}.
\end{rem*}
\begin{thm}
If the process $\big\{ X_{n},Y_{n}\big\}_{n=1}^{\infty}$ is stationary,
then there exists a random constant strategy $\big(b^{*}\big)$ that
grows asymptotically as fast as the log-optimal strategy $\big(b_{n}^{*}\big)$
as:\label{Theorem 2}
\[
\lim_{n\to\infty}\big(W_{n}\big(b^{*}\big)-W_{n}\big(b_{n}^{*}\big)\big)=0,\text{ a.s}.
\]
This follows from the limits holding in terms of normalized assets'
returns:
\begin{align*}
\lim_{n\to\infty}\dfrac{1}{n}\sum_{i=1}^{n}\log\big<b^{*},U_{i}\big>=\lim_{n\to\infty}\dfrac{1}{n}\sum_{i=1}^{n}\log\big<b_{i}^{*},U_{i}\big> & =\mathbb{E}\big(\log\big<b_{\infty}^{*},U_{1}\big>|\mathcal{I}\big)\\
 & =\max_{b\in\mathcal{B}^{m}}\mathbb{E}\big(\log\big<b,U_{1}\big>|\mathcal{I}\big),\text{ a.s},
\end{align*}
where the portfolio $b_{\infty}^{*}$ is a measurable maximizer for
$\max_{b\in\mathcal{B}^{m}}\mathbb{E}\big(\log\left\langle b,U_{1}\right\rangle |\mathcal{F}_{\infty}^{X}\big)$.
Moreover, if the market does not have trash assets, then the convergences
above are also in $L^{1}$; and if the market is also ergodic, then
the optimal constant strategy is non-random.
\end{thm}
\begin{proof}
By Lemma \ref{Lemma 2}, an inequality $\limsup_{n\to\infty}\big(W_{n}\big(b^{*}\big)-W_{n}\big(b_{n}^{*}\big)\big)\leq0$
is established almost surely for any constant strategy $\big(b^{*}\big)$.
Therefore, by using Lemma \ref{Lemma 1}, it is sufficient to prove
the following inequality for some constant strategy $\big(b^{*}\big)$:
\[
\liminf_{n\to\infty}\big(W_{n}\big(b^{*}\big)-W_{n}\big(b_{n}^{*}\big)\big)=\liminf_{n\to\infty}\dfrac{1}{n}\Big(\sum_{i=1}^{n}\log\big<b^{*},U_{i}\big>-\sum_{i=1}^{n}\log\big<b_{i}^{*},U_{i}\big>\Big)\geq0,\text{ a.s.}
\]
Then, the remaining limits in the statement will follow.\smallskip 

First, we prove the condition of being $L^{1}$-dominated for any
sequence of logarithmic returns of a strategy $\big(b_{n}\big)$ with
the measurable portfolios $b_{n}$ for all $n$, i.e., $\mathbb{E}\big(\sup_{n}|\log\big<b_{n},U_{n}\big>|\big)$
is finite. This is required to invoke Breiman's generalized ergodic
theorem (noting that the original ergodic theorem by \citet{Breiman1957}
only modified the almost sure convergence from Birkhoff's ergodic
theorem, which was later extended to convergence in $L^{1}$ by \citet{Algoet1994},
with the same condition of integrability). Specifically, we have:
\begin{equation}
|\log\big<b_{n},U_{n}\big>|\leq\max\left\{ |\log\big(c\big)|,\,|\log\big(m\big)|\right\} <\infty,\text{ }\forall n,\text{ a.s.}\label{uniformly integrable}
\end{equation}
This is due to the bounds $\log\big<b_{n},U_{n}\big>\leq\log\big(m\big)$
by Lemma \ref{Lemma 1}, and $\log\big(c\big)\leq\log\big<b_{n},U_{n}\big>$
by the market condition without trash assets, as follows:
\[
U_{n,k}\geq\min_{1\leq i\leq m}\frac{mX_{n,i}}{\sum_{j=1}^{m}X_{n,j}}\geq\min_{1\leq i\leq m}\max_{1\leq j\leq m}\frac{X_{n,i}}{X_{n,j}}\geq c,\text{ }\forall n,\forall k\in\left\{ 1,...,m\right\} ,\text{ a.s}.
\]
Next, we proceed to establish the limits, in terms of the normalized
assets' returns, of the log-optimal strategy and a random constant
strategy as follows.\smallskip 

\emph{(The log-optimal strategy)}. Consider a log-optimal strategy
$\big(b_{n}^{*}\big)$ with the $\mathcal{F}_{n}^{X}$-measurable
portfolios $b_{n}^{*}$ for all $n$, which can be selected by Lemma
\ref{Lemma 3}. By stationarity, for each $n$, the measurable portfolio
$b_{n}^{*}$ can be expressed as a shift from another portfolio:
\[
b_{n}^{*}\big(\omega\big)\eqqcolon b_{n}^{-*}\big(T^{(n-1)}\omega\big)\text{, where }\mathbb{E}\big(\log\big<b_{n}^{-*},U_{1}\big>|T^{(1-n)}\mathcal{F}_{n}^{X}\big)=\max_{b\in\mathcal{B}^{m}}\mathbb{E}\big(\log\big<b,U_{1}\big>|T^{(1-n)}\mathcal{F}_{n}^{X}\big).
\]
The inverse filtration $\left\{ T^{(1-n)}\mathcal{F}_{n}^{X}\right\} _{n=1}^{\infty}$
increases up to the following inverse limiting sub-$\sigma$-field:
\[
\mathcal{F}_{-\infty}^{X}\coloneqq T^{-\infty}\mathcal{F}_{\infty}^{X}=\sigma\big(\cup_{n=1}^{\infty}T^{(1-n)}\mathcal{F}_{n}^{X}\big)\subseteq\mathbb{F}.
\]
By Lemma \ref{Lemma 4}, we then have $\big<b_{n}^{-*},U_{1}\big>\to\big<b_{\infty}^{-*},U_{1}\big>$
almost surely, where the portfolio $b_{\infty}^{-*}$ is a maximizer
for $\max_{b\in\mathcal{B}^{m}}\mathbb{E}\big(\log\left\langle b,U_{1}\right\rangle |\mathcal{F}_{-\infty}^{X}\big)$,
which can also be $\mathcal{F}_{-\infty}^{X}$-measurably selected
by Lemma \ref{Lemma 3}. Given the invariant $\sigma$-field $\mathcal{I}$,
Breiman's generalized ergodic theorem yields:
\begin{align}
\lim_{n\to\infty}\dfrac{1}{n}\sum_{i=1}^{n}\log\big<b_{i}^{*},U_{i}\big>\big(\omega\big) & =\lim_{n\to\infty}\dfrac{1}{n}\sum_{i=1}^{n}\log\big<b_{i}^{-*}\big(T^{(i-1)}\omega\big),U_{1}\big(T^{(i-1)}\omega\big)\big>\nonumber \\
 & =\mathbb{E}\big(\log\big<b_{\infty}^{-*},U_{1}\big>|\mathcal{I}\big)\nonumber \\
 & =\mathbb{E}\big(\log\big<b_{\infty}^{*},U_{1}\big>|\mathcal{I}\big),\text{ a.s and in }L^{1},\label{Limit 1}
\end{align}
where $b_{\infty}^{*}$ is a $\mathcal{F}_{\infty}^{X}$-measurable
maximizer for $\max_{b\in\mathcal{B}^{m}}\mathbb{E}\big(\log\left\langle b,U_{1}\right\rangle |\mathcal{F}_{\infty}^{X}\big)$,
and the $L^{1}$-domination condition is satisfied by the sequence
$\left\{ \log\big<b_{n}^{*},U_{n}\big>\right\} _{n=1}^{\infty}$,
as stated earlier.\smallskip 

\emph{(The random constant strategy)}. Consider a random optimal portfolio
$b^{*}$ as a maximizer for $\max_{b\in\mathcal{B}^{m}}\mathbb{E}\big(\log\big<b,U_{1}\big>|\mathcal{I}\big)$,
which forms the corresponding random constant strategy $\big(b^{*}\big)$.
Given the well-defined $\mathbb{E}\big(\log\big<b^{*},U_{1}\big>|\mathcal{I}\big)$
by Lemma \ref{Lemma 1}, we have:
\begin{equation}
\lim_{n\to\infty}\dfrac{1}{n}\sum_{i=1}^{n}\log\big<b^{*},U_{i}\big>=\mathbb{E}\big(\log\big<b^{*},U_{1}\big>|\mathcal{I}\big)=\max_{b\in\mathcal{B}^{m}}\mathbb{E}\big(\log\big<b,U_{1}\big>|\mathcal{I}\big),\text{ a.s and in }L^{1},\label{Limit 2}
\end{equation}
by invoking Birkhoff's ergodic theorem. Noting that the $L^{1}$-convergence
follows from the almost sure convergence by applying Lebesgue's dominated
convergence theorem, given that the process $\big\{ n^{-1}\sum_{i=1}^{n}\log\big<b^{*},U_{i}\big>\big\}_{n=1}^{\infty}$
is uniformly integrable, as established by the bound in (\ref{uniformly integrable}).\smallskip 

\emph{(Equalities derivation)}. We now establish the equality between
the obtained limits in (\ref{Limit 1}) and (\ref{Limit 2}). First,
recall the following equality:
\[
\lim_{n\to\infty}\big(W_{n}\big(b^{\omega}\big)-W_{n}\big(b_{n}^{*}\big)\big)=\lim_{n\to\infty}\frac{1}{n}\sum_{i=1}^{n}\big(\log\big<b^{\omega},U_{i}\big>-\log\big<b_{i}^{*},U_{i}\big>\big)=0,\text{ a.s,}
\]
for a random constant strategy $\big(b^{\omega}\big)$ by Theorem
\ref{Theorem 1}. As a result, we have:
\begin{align}
\liminf_{n\to\infty}\frac{1}{n}\sum_{i=1}^{n}\big(\log\big<b^{\omega},U_{i}\big>-\log\big<b_{i}^{*},U_{i}\big>\big) & \leq\liminf_{n\to\infty}\frac{1}{n}\Big(\max_{b\in\mathcal{B}^{m}}\sum_{i=1}^{n}\log\big<b,U_{i}\big>-\sum_{i=1}^{n}\log\big<b_{i}^{*},U_{i}\big>\Big)\nonumber \\
 & =\max_{b\in\mathcal{B}^{m}}\mathbb{E}\big(\log\big<b,U_{1}\big>|\mathcal{I}\big)-\lim_{n\to\infty}\frac{1}{n}\sum_{i=1}^{n}\log\big<b_{i}^{*},U_{i}\big>\nonumber \\
 & =\lim_{n\to\infty}\frac{1}{n}\Big(\sum_{i=1}^{n}\log\big<b^{*},U_{i}\big>-\sum_{i=1}^{n}\log\big<b_{i}^{*},U_{i}\big>\Big)\text{, a.s}.\label{Limit equality}
\end{align}
The establishment of (\ref{Limit equality}) above follows from the
limit in (\ref{Limit 2}) and the following:
\begin{align*}
\lim_{n\to\infty}\max_{b\in\mathcal{B}^{m}}\frac{1}{n}\sum_{i=1}^{n}\log\big<b,U_{i}\big> & =\lim_{n\to\infty}\max_{b\in\mathcal{B}^{m}}\mathbb{E}^{P_{n}}\big(\log\big<b,U\big>\big)\\
 & =\max_{b\in\mathcal{B}^{m}}\mathbb{E}\big(\log\big<b,U_{1}\big>|\mathcal{I}\big),\text{ a.s,}
\end{align*}
which is due to Lemma \ref{Lemma 4} regarding the sequence of distributions
$\big\{ P_{n}\big\}_{n=1}^{\infty}$ defined as:
\[
P_{n}\big(A\big)\coloneqq\dfrac{1}{n}\sum_{i=1}^{n}\mathbb{I}_{U_{i}}\big(A\big),\,\forall A\subseteq\mathcal{U},\forall n,
\]
with the notation $\mathbb{E}^{P_{n}}(\cdot)$ denoting the expected
value with respect to $U\sim P_{n}$. Importantly, noting that the
sequence of empirical distributions $\big\{ P_{n}\big\}_{n=1}^{\infty}$
converges weakly almost surely to the limiting distribution $\mathbb{P}\big(U_{1}|\mathcal{I}\big)$
by the argument below:
\[
\lim_{n\to\infty}\mathbb{E}^{P_{n}}\big(f\big(U\big)\big)=\lim_{n\to\infty}\frac{1}{n}\sum_{i=1}^{n}f\big(U_{i}\big)=\mathbb{E}\big(f\big(U_{1}\big)|\mathcal{I}\big),\text{ a.s},
\]
for any bounded continuous function $f(\cdot)$ by Birkhoff's ergodic
theorem.\smallskip 

Therefore, we finally obtain the following needed almost sure equality:
\begin{align*}
0 & =\lim_{n\to\infty}\big(W_{n}\big(b^{*}\big)-W_{n}\big(b_{n}^{*}\big)\big)\\
 & =\max_{b\in\mathcal{B}^{m}}\mathbb{E}\big(\log\big<b,U_{1}\big>|\mathcal{I}\big)-\mathbb{E}\big(\log\big<b_{-\infty}^{*},U_{1}\big>|\mathcal{I}\big),\text{ a.s,}
\end{align*}
which, along with the $L^{1}$-convergence in (\ref{Limit 1}) and
(\ref{Limit 2}), results in the following:
\begin{align*}
0= & \lim_{n\to\infty}\mathbb{E}\Big(\big|\frac{1}{n}\sum_{i=1}^{n}\log\big<b^{*},U_{i}\big>-\mathbb{E}\big(\log\big<b_{\infty}^{*},U_{1}\big>|\mathcal{I}\big)\big|\Big)\\
 & +\lim_{n\to\infty}\mathbb{E}\Big(\big|\dfrac{1}{n}\sum_{i=1}^{n}\log\big<b_{i}^{*},U_{i}\big>-\mathbb{E}\big(\log\big<b_{\infty}^{*},U_{1}\big>|\mathcal{I}\big)\big|\Big)\\
\geq & \lim_{n\to\infty}\mathbb{E}\big(|W_{n}\big(b^{*}\big)-W_{n}\big(b_{n}^{*}\big)|\big).
\end{align*}
Hence, the last $L^{1}$-convergence is also obtained, and the proof
is completely finished.
\end{proof}
According to Theorem \ref{Theorem 2}, although the best possible
growth rate of all strategies in the market might not converge, or
might even be infinite, there exists an unsophisticated constant strategy
that is possibly not exceeded by any optimally sophisticated strategy,
leveraged by knowledge of all market information, including unobservable
ones, by a fixed magnitude infinitely often as time elapses. Specifically,
recalling that by changing the variable to the normalized assets'
returns, the growth rate of any strategy $\big(b_{n}\big)$ is decomposed
into two distinct terms as follows:
\[
W_{n}\big(b_{n}\big)=\frac{1}{n}\mathbb{E}\big(\sum_{i=1}^{n}\log\big<b_{i},U_{i}\big>\big)+\frac{1}{n}\mathbb{E}\big(\sum_{i=1}^{n}\log\big<\hat{b},X_{i}\big>\big),
\]
where the optimal rate of the first term converges almost surely to
a finite limit as $n\to\infty$ (without additional assumptions),
as asserted by Theorem \ref{Theorem 2}. In contrast, the growth rate
$W_{n}\big(\hat{b}\big)$ associated with the reference constant strategy
$\big(\hat{b}\big)$ may diverge or oscillate infinitely, as it depends
on the uncontrolled original variables $X_{n}$. Therefore, it is
clear that the original growth rate $W_{n}\big(b_{n}^{*}\big)$ of
the log-optimal strategy $\big(b_{n}^{*}\big)$ will converge almost
surely to a finite limit if the growth rate $W_{n}\big(\hat{b}\big)$
has a finite limit. In particular, if the market is also ergodic,
the following holds by Lemma \ref{Lemma 3}:
\begin{align}
\mathbb{E}\big(\max_{b\in\mathcal{B}^{m}}\mathbb{E}\big(\log\big<b,U_{1}\big>|\mathcal{F}_{\infty}^{X}\big) & =\mathbb{E}\big(\log\big<b_{\infty}^{*},U_{1}\big>\big)\nonumber \\
 & =\max_{b\in\mathcal{B}^{m}}\mathbb{E}\big(\log\big<b,U_{1}\big>\big),\label{ergodic case equality}
\end{align}
so the optimal constant strategy $\big(b^{*}\big)$ is a non-random
maximizer for $\max_{b\in\mathcal{B}^{m}}\mathbb{E}\big(\log\big<b,U_{1}\big>\big)$.\smallskip 

Furthermore, the expected difference in growth rates between the log-optimal
strategy and the random optimal constant strategy, which becomes non-random
if the ergodicity condition is introduced, is also depleted to zero,
as established in Corollary \ref{Corollary 1}. This surprising result
highlights that such an optimal constant strategy is comparable to
any other dynamic strategy among all the accessible ones that asymptotically
maximize the expected growth rate of strategy, which might seem counterintuitive.
Finally, in the case of a stationary and ergodic process $\big\{ X_{n},Y_{n}\big\}_{n=1}^{\infty}$,
the limits of the growth rate and the expected growth rate of the
regarded optimal strategies, as established in the corollary, are
the same and equal to the terms in (\ref{ergodic case equality}).
\begin{cor}
If the process $\big\{ X_{n},Y_{n}\big\}_{n=1}^{\infty}$ is stationary
(the condition of a no-trash assets market is not necessary), then
the following limits hold for the log-optimal strategy $\big(b_{n}^{*}\big)$
and a random constant strategy $\big(b^{*}\big)$ among all possible
strategies $\big(b_{n}\big)$ as:\label{Corollary 1}
\begin{align*}
\lim_{k\to\infty}\max_{\left(b_{n}\right)}\mathbb{E}\big(W_{k}\big(b_{k}\big)-W_{k}\big(b^{*}\big)\big)=\lim_{k\to\infty}\mathbb{E}\big(W_{k}\big(b_{k}^{*}\big)-W_{k}\big(b^{*}\big)\big) & =0,
\end{align*}
where the optimal expected growth rate in terms of normalized assets'
returns has a limit:
\begin{align*}
\lim_{k\to\infty}\max_{\left(b_{n}\right)}\frac{1}{k}\mathbb{E}\big(\sum_{i=1}^{k}\log\big<b_{i},U_{i}\big>\big)=\lim_{k\to\infty}\frac{1}{k}\mathbb{E}\big(\sum_{i=1}^{k}\log\big<b_{i}^{*},U_{i}\big>\big) & =\lim_{k\to\infty}\frac{1}{k}\mathbb{E}\big(\sum_{i=1}^{k}\log\big<b^{*},U_{i}\big>\big)\\
 & =\mathbb{E}\big(\max_{b\in\mathcal{B}^{m}}\mathbb{E}\big(\log\big<b,U_{1}\big>|\mathcal{F}_{\infty}^{X}\big)\big)\\
 & =\mathbb{E}\big(\log\big<b^{*},U_{1}\big>\big).
\end{align*}
\end{cor}
\begin{proof}
If the market does not have trash assets, the proof is straightforward
by applying Jensen's inequality for the absolute value function and
invoking the $L^{1}$-convergence in Theorem \ref{Theorem 2}. We
then have the following limit, given the random optimal constant strategy
$\big(b^{*}\big)$ as in (\ref{Limit 2}):
\begin{align*}
\lim_{k\to\infty}|\max_{\left(b_{n}\right)}\mathbb{E}\big(W_{k}\big(b_{k}\big)-W_{k}\big(b^{*}\big)\big)| & =\lim_{k\to\infty}|\mathbb{E}\big(W_{k}\big(b_{k}^{*}\big)-W_{k}\big(b^{*}\big)\big)|\\
 & \leq\lim_{k\to\infty}\mathbb{E}\big(|W_{k}\big(b_{k}^{*}\big)-W_{k}\big(b^{*}\big)|\big)=0.
\end{align*}
By similar arguments, the limiting optimal expected growth rate of
the best strategy among all possible ones, in terms of normalized
assets' returns, is a direct consequence of the $L^{1}$-convergence
in (\ref{Limit 1}) and (\ref{Limit 2}), which implies the convergence
of expected values.\smallskip 

Now, consider the general case without the no-trash assets condition.
We proceed to prove the needed convergence $\mathbb{E}\big(W_{n}\big(b_{n}^{*}\big)-W_{n}\big(b^{*}\big)\big)\to0$
by establishing the equality between the individual expected growth
rates in terms of normalized assets' returns of the involved optimal
strategies. Specifically, consider the random optimal constant strategy
$\big(b^{*}\big)$ as in (\ref{Limit 2}). We have:
\[
\lim_{n\to\infty}\frac{1}{n}\mathbb{E}\big(\sum_{i=1}^{n}\log\big<b^{*},U_{i}\big>\big)=\lim_{n\to\infty}\frac{1}{n}\sum_{i=1}^{n}\mathbb{E}\big(\log\big<b^{*},U_{i}\big>\big)=\mathbb{E}\big(\log\big<b^{*},U_{1}\big>\big),
\]
because the random portfolio $b^{*}\big(\omega\big)$ is constantly
chosen depending on each $\omega\in\Omega$ over all random pairs
$\big(X_{n},Y_{n}\big)\big(\omega\big)$ for all $n$, of the stationary
process $\big\{ X_{n},Y_{n}\big\}_{n=1}^{\infty}$.\smallskip 

Moreover, by taking the expectation of both sides of the almost sure
equalities of the limiting growth rates of the two regarded optimal
strategies in Theorem \ref{Theorem 2}, and then applying Lemma \ref{Lemma 3}
on each of them (noting that $b^{*}$ can be treated as a log-optimal
portfolio that is measurably selected depending on the $\sigma$-field
$\mathcal{I}$), we obtain the following equality:
\begin{align*}
\mathbb{E}\big(\log\big<b^{*},U_{1}\big>\big)=\mathbb{E}\big(\max_{b\in\mathcal{B}^{m}}\mathbb{E}\big(\log\big<b,U_{1}\big>|\mathcal{I}\big)\big) & =\mathbb{E}\big(\max_{b\in\mathcal{B}^{m}}\mathbb{E}\big(\log\big<b,U_{1}\big>|\mathcal{F}_{\infty}^{X}\big)\big),
\end{align*}
where the expected value on the right-hand side is also the limiting
expected growth rate, in terms of normalized assets' returns, of the
log-optimal strategy, as shown in the following arguments.\smallskip 

Let us recall that $\mathbb{P}\big(U_{1}|\mathcal{F}_{n}^{X}\big)\to\mathbb{P}\big(U_{1}|\mathcal{F}_{\infty}^{X}\big)$
weakly almost surely, according to Lemma \ref{Lemma 4}, so $\max_{b\in\mathcal{B}^{m}}\mathbb{E}\big(\log\big<b,U_{n}\big>|\mathcal{F}_{n}^{X}\big)$$\to\max_{b\in\mathcal{B}^{m}}\mathbb{E}\big(\log\big<b,U_{1}\big>|\mathcal{F}_{\infty}^{X}\big)$
by the stationarity of the market process. Subsequently, by using
Lebesgue's dominated convergence theorem (since the maximal conditional
expected values $\max_{b\in\mathcal{B}^{m}}\mathbb{E}\big(\log\big<b,U_{n}\big>|\mathcal{F}_{n}^{X}\big)$
are non-negative and upper bounded for all $n$ by Lemma \ref{Lemma 1}),
we obtain the following convergence:
\begin{align*}
\lim_{n\to\infty}\mathbb{E}\big(\max_{b\in\mathcal{B}^{m}}\mathbb{E}\left(\log\left\langle b,U_{n}\right\rangle |\mathcal{F}_{n}^{X}\right)\big) & =\mathbb{E}\big(\lim_{n\to\infty}\max_{b\in\mathcal{B}^{m}}\mathbb{E}\left(\log\left\langle b,U_{n}\right\rangle |\mathcal{F}_{n}^{X}\right)\big)\\
 & =\mathbb{E}\big(\max_{b\in\mathcal{B}^{m}}\mathbb{E}\left(\log\left\langle b,U_{1}\right\rangle |\mathcal{F}_{\infty}^{X}\right)\big).
\end{align*}
Moreover, by further using the chain rule in (\ref{chain rule}) and
Cesaro's mean theorem, we have:
\begin{align*}
\lim_{k\to\infty}\frac{1}{k}\mathbb{E}\big(\sum_{i=1}^{k}\log\big<b_{i}^{*},U_{i}\big>\big) & =\lim_{k\to\infty}\max_{\left(b_{n}\right)}\frac{1}{k}\mathbb{E}\big(\sum_{i=1}^{k}\log\big<b_{i},U_{i}\big>\big)\\
 & =\lim_{k\to\infty}\dfrac{1}{k}\sum_{i=1}^{k}\mathbb{E}\big(\max_{b\in\mathcal{B}^{m}}\mathbb{E}\big(\log\big<b,U_{i}\big>|\mathcal{F}_{i}^{X}\big)\big)\\
 & =\lim_{k\to\infty}\mathbb{E}\big(\max_{b\in\mathcal{B}^{m}}\mathbb{E}\big(\log\big<b,U_{k}\big>|\mathcal{F}_{k}^{X}\big)\big)\\
 & =\mathbb{E}\big(\max_{b\in\mathcal{B}^{m}}\mathbb{E}\big(\log\big<b,U_{1}\big>|\mathcal{F}_{\infty}^{X}\big)\big).
\end{align*}
Therefore, we finally obtain all the necessary equalities of the statement
and complete the proof.
\end{proof}

\subsection{Market with latent dependence structures}

In this section, we discuss market scenarios with latent dependence
structures, which can be derived as consequences of Theorem \ref{Theorem 1}
and Theorem \ref{Theorem 2}. This serves to demonstrate the generality
of the modeling settings and the properties obtained in the previous
sections for capturing various real-market behaviors. In a real market,
such dependence is justified, as economic agents and institutions
typically react to market events, which may be correlated within the
same market or across different markets due to causal effects. The
structure of dependence may vary due to factors such as technological
advancements, regulations, and social features, which influence the
spread, speed of reception, and perception of information across markets.
It is important to note that the dependence structures in markets
are often not fully known, as not all supplementary information is
publicly observable or may be prohibitively costly to collect and
store. Additionally, some information may not be truly valuable due
to weak correlations with asset returns. We now proceed to capture
the dependence structure through the following three market phenomena.\smallskip 

\textbf{Market with non-decaying impact of information}. Consider
a market process $\big\{ X_{n},Y_{n}\big\}_{n=1}^{\infty}$, where
the current assets' returns $X_{n}$ and other market features $Y_{n}$,
including unobservable ones, depend on past information from the start
to the previous period $n-1$. In other words, the market has unlimited
memory of past events, such that their effects can persist infinitely,
with non-decaying impacts on the future market over an unrestricted
horizon. This phenomenon of unlimited memory can be modeled by a market
process with a transition property, where the distribution of each
pairwise random variable depends on all past events, or alternatively,
it is adapted to the filtration $\left\{ \mathcal{F}_{n}\right\} _{n=1}^{\infty}$
that encompasses increasing past market information. Proposition \ref{Proposition 1}
considers this type of market model with a mild assumption, enabling
the result to be directly derived without the need for stationarity
in the market process.\smallskip 
\begin{rem*}
The introduction of the market model in Proposition \ref{Proposition 1}
is intended as a specific instance for modeling the non-decaying impact
of information and is not central to the main discussion in this paper.
However, it provides a weaker assumption than strict stationarity,
under which the existence of a random optimal constant strategy remains
valid, even though the growth rate in terms of normalized assets'
returns of the log-optimal strategy $n^{-1}\sum_{i=1}^{n}\log\big<b_{i}^{*},U_{i}\big>$
might not be guaranteed. Another useful aspect is that it provides
an alternative proof for Theorem \ref{Theorem 1}, without the need
for Lemma \ref{Lemma 5} and the discrete market condition as specified
in Definition \ref{Def-of-discrete-market}.
\end{rem*}
\begin{prop}
Consider the market process $\big\{ X_{n},Y_{n}\big\}_{n=1}^{\infty}$,
where the distributions of the random pairs $\big(X_{n},Y_{n}\big)$
depend on past information, such that $\big(X_{n},Y_{n}\big)\sim\mathbb{P}\big(X_{1},Y_{1}|\mathcal{F}_{n}\big)$
for all $n\geq2$, where $\mathbb{P}\big(X_{1},Y_{1}|\mathcal{F}_{n}\big)$
is the regular conditional distribution given the $\sigma$-field
$\mathcal{F}_{n}$ admitted by $\big(X_{1},Y_{1}\big)$. Then, there
exists a random optimal constant strategy $\big(b^{\omega}\big)$
such that $\big(W_{n}\big(b^{\omega}\big)-W_{n}\big(b_{n}^{*}\big)\big)\to0$
almost surely, which is also in $L^{1}$ under the additional no-trash
assets condition.\label{Proposition 1}
\end{prop}
\begin{proof}
Based on Lemma \ref{Lemma 4}, if the space of all distributions of
normalized assets' returns $U_{1}$ is equipped with the weak topology,
then for two distributions $\hat{Q}$ and $\tilde{Q}$ that are close
to each other in the weak sense, i.e., their corresponding distance
of expected values satisfies $|\mathbb{E}^{\hat{Q}}\big(f\big(U_{1}\big)\big)-\mathbb{E}^{\tilde{Q}}\big(f\big(U_{1}\big)\big)|<\epsilon$
for arbitrarily small $\epsilon>0$ and any bounded continuous real
function $f(\cdot)$, the portfolios $b^{*}\big(\hat{Q}\big)$ maximizing
$\mathbb{E}^{\hat{Q}}\big(\log\big<b,U_{1}\big>\big)$ and $b^{*}\big(\tilde{Q}\big)$
maximizing $\mathbb{E}^{\tilde{Q}}\big(\log\big<b,U_{1}\big>\big)$
must satisfy $\big(\log\big<b^{*}\big(\hat{Q}\big),U_{1}\big>-\log\big<b^{*}\big(\tilde{Q}\big),U_{1}\big>\big)\to0$
as $\epsilon\to0$ for $\hat{Q}$-almost all and also $\tilde{Q}$-almost
all $U_{1}$. Therefore, by the assumption for the market process,
since the regular conditional distribution $\mathbb{P}\big(U_{1}|\mathcal{F}_{n}^{X}\big)$
converges weakly almost surely to the limiting distribution $\mathbb{P}\big(U_{1}|\mathcal{F}_{\infty}^{X}\big)$
by Lemma \ref{Lemma 4}, we have $\big(\log\big<b_{n}^{*},U_{n}\big>-\log\big<b_{\infty}^{*},U_{n}\big>\big)\to0$
for $\mathbb{P}\big(U_{1}|\mathcal{F}_{n}^{X}\big)$-almost all $U_{n}$,
almost surely. Thus, the desired almost sure convergence $\big(W_{n}\big(b^{\omega}\big)-W_{n}\big(b_{n}^{*}\big)\big)\to0$
follows with the chosen $b^{\omega}\coloneqq b_{\infty}^{*}$ by applying
the Cesaro mean theorem as in the proof of Theorem \ref{Theorem 1}.
Additionally, the $L^{1}$ convergence can be easily derived using
the Lebesgue's dominated convergence theorem.\smallskip

Alternatively, the almost sure convergence $\big(W_{n}\big(b^{\omega}\big)-W_{n}\big(b_{n}^{*}\big)\big)\to0$
can be demonstrated more clearly than previously described by adapting
the arguments from \citet{Lam2024} (Lemma 1 and Lemma 2) to the portfolio\textquoteright s
log-return. Specifically, we sketch the proof as follows: Define $\left\{ b^{*}\big(Q\big)\right\} $
as the set of all maximizers for the corresponding $\mathbb{E}^{Q}\big(\log\big<b,U_{1}\big>\big)$.
Since the regular conditional distributions $\mathbb{P}\big(U_{n}|\mathcal{F}_{n}^{X}\big)\to\mathbb{P}\big(U_{1}|\mathcal{F}_{\infty}^{X}\big)$,
the Hausdorff distance between the involved sets, denoted by $d_{H}\big(b^{*}\big(\mathbb{P}\big(U_{n}|\mathcal{F}_{n}^{X}\big)\big),b^{*}\big(\mathbb{P}\big(U_{1}|\mathcal{F}_{\infty}^{X}\big)\big)\big)$,
must converge to zero as $n\to\infty$, where the involved sets are
convex and compact. Therefore, for each maximizer $b_{\infty}^{*}$,
there always exists a maximizer $b_{n}^{*}$ such that $\big(b_{n}^{*}-b_{\infty}^{*}\big)\to0$,
so $\big(\log\big<b_{n}^{*},U_{n}\big>-\log\big<b_{\infty}^{*},U_{n}\big>\big)\to0$
for $\mathbb{P}\big(U_{1}|\mathcal{F}_{n}^{X}\big)$-almost all $U_{n}$
(noting critically that all $U_{n}$ are bounded, as in (\ref{U is bounded})),
which holds for any $b_{n}^{*}$.
\end{proof}
\textbf{Market with decaying impact of information}. Although a real
market often reflects memories of the past, sometimes the impact of
past events does not last very long on future market developments,
which is probably caused by the limited memory of agents and institutions
in the market. In this kind of market, future assets' returns do not
immediately incorporate all available information but rather gradually
absorb it over time until the past impact has no influence. A market
process with a Markov property of finite order should be well-suited
for modeling the persistence of past impacts on the future. Let us
consider the market process $\big\{ X_{n},Y_{n}\big\}_{n=1}^{\infty}$
as a Markov chain of order $h\geq1$, so that the variable $\big(X_{n},Y_{n}\big)$
depends on the previous $h$ variables, and its conditional distribution
given the $\sigma$-field $\mathcal{F}_{n}$ is:
\begin{equation}
\mathbb{P}\big(X_{n},Y_{n}|\mathcal{F}_{n}\big)=\mathbb{P}\big(X_{n},Y_{n}|\sigma\big(X_{n-h}^{n-1},Y_{n-h}^{n-1}\big)\big),\,\forall n>h,\label{Markov property}
\end{equation}
where the order $h$ represents the length of the market memory of
past information. The length of the market memory is higher when the
impact of past information decays at a slower rate, i.e., it is more
persistent. If the market is also stationary, then Theorem \ref{Theorem 2}
reduces to Corollary \ref{Corollary 2}.\smallskip

\textbf{Memoryless or past-independent market}. In the case of the
market having no memory of past events, the side information available
at the same period, revealed right before the realization of assets'
returns, is the only factor affecting the assets' returns. This means
the market is highly efficient in the sense that it absorbs all available
side information immediately and leaves no capacity for past information
to impact future market developments. The memoryless nature of the
market could be modeled as a special case of the market with a Markov
property of order $h=0$ in (\ref{Markov property}), which implies
the (joint) distributions of the variables $\big(X_{n},Y_{n}\big)$
are independent of each other. Moreover, as in the case of a Markov
process with finite order, Theorem \ref{Theorem 2} also reduces to
Corollary \ref{Corollary 2} if the distributions of the variables
are identical.
\begin{cor}
Consider a process $\big\{ X_{n},Y_{n}\big\}_{n=1}^{\infty}$ which
is either identically and independently distributed or a stationary
Markov chain of any finite order $h\geq1$ according to (\ref{Markov property}).
Then, there exists either a non-random or random optimal constant
strategy $\big(b^{*}\big)$, respectively, such that $\big(W_{n}\big(b^{*}\big)-W_{n}\big(b_{n}^{*}\big)\big)\to0$
almost surely, which is also in $L^{1}$ under the no-trash assets
condition.\label{Corollary 2}
\end{cor}
\begin{proof}
The result is simply a direct consequence of Theorem \ref{Theorem 2},
taking into account that an i.i.d. process $\big\{ X_{n},Y_{n}\big\}_{n=1}^{\infty}$
is also stationary and ergodic, since the only invariant $\sigma$-field
is the trivial one, i.e., $\mathcal{I}=\left\{ \Omega,\emptyset\right\} $,
so the constant optimal strategy $\big(b^{*}\big)$ is non-random.
\end{proof}
According to Corollary \ref{Corollary 2}, if the assets' returns
depend solely on the side market information such that the random
pairs are i.i.d., the non-random optimal constant strategy coincides
with the log-optimal strategy, which is selected as a maximizer for
$\max_{b\in\mathcal{B}^{m}}\mathbb{E}\big(\log\big<b,U_{1}\big>\big)$.
This result extends the well-known theorem for the i.i.d. process
of sole assets' returns $\left\{ X_{n}\right\} _{n=1}^{\infty}$ in
the literature, which asserts the existence of an optimal constant
strategy achieving the maximal limiting growth rate. Intuitively,
the extended result means that the observation of i.i.d. side information,
which affects the periodic assets' returns in a time-invariant dependency,
almost surely does not improve decision-making over being unaware
of the side information and dependence structure, in terms of growth
rate over time; thus, even unobservable market features do not matter.
However, if the side information is i.i.d., but the dependence structure
between them and the assets' returns changes over time, then the conditional
distribution $\mathbb{P}\big(X_{n}|\sigma\big(Y_{n}\big)\big)$ varies
with $n$, and the results of Corollary \ref{Corollary 2} are no
longer valid due to the violated assumption.\smallskip 
\begin{rem*}
In recent literature, \citet{Bhatt2023} propose a strategy construction
using fully observable i.i.d. side market information, which can be
considered from a different perspective for comparison. Additionally,
\citet{Cuchiero2019} investigate a stationary and ergodic Markov
chain (order $h=1$) for the market process of only assets' returns,
a sub-case of our model without side information. The authors considers
the space of functions depending on assets' returns $X$, yielding
a portfolio $\text{b\ensuremath{\in}}\mathcal{B}^{m}$, where there
exists an optimal function generating a strategy with the maximal
limiting growth rate, depending on the process $\left\{ X_{n}\right\} _{n=1}^{\infty}$.
They then propose a strategy by mixing all the functionally generated
portfolios over time according to a fixed initial distribution on
the function space, such that it approaches the maximal limiting growth
rate of the strategy generated by the optimal function. However, this
method of function mixing seems impractical. A closer inspection reveals
that the log-optimal portfolios, defined as $\operatorname*{argmax}_{b\in\mathcal{B}^{m}}\mathbb{E}\big(\log\big<b,U_{1}\big>|\sigma\big(X_{0}\big)\big)$,
are optimal in the considered space of functions, as shown by Corollary
\ref{Corollary 2}. Moreover, the next section presents a learning
algorithm that provides a practical way of mixing functions to create
a strategy comparable to the log-optimal one for any finite-order
($h\geq1$) Markov or i.i.d. market process with side information.
\end{rem*}

\section{Learning algorithms under unknown distributions}

Although Theorem \ref{Theorem 2} shows the existence of a random
(or non-random, with additional ergodicity) optimal constant strategy,
which does not rely on market events realized over time as the log-optimal
strategy does, or on being conditioned on the infinitely past events,
as the random optimal constant strategy in Theorem \ref{Theorem 1},
it still depends on the knowledge of the conditional distribution
based on the invariant events of the market process. As discussed
in the previous section, the true ergodic mode of the stationary market
process cannot be identified by observing any finite past information
rather than the infinite past, thus an optimal constant strategy is
also theoretically unidentified with any finite time horizon. Meanwhile,
although the log-optimal strategy could detect the ergodic mode of
the process, it requires knowing the infinite-dimensional distribution.
However, the established existence of an optimal constant strategy
clearly helps reduce the challenge of knowing such an infinite-dimensional
distribution, which requires full observability of all market information,
to a much easier and more feasible task. Therefore, in this section,
two practical learning methods of algorithms are compared: one replicating
the optimal constant strategy and the other replicating the log-optimal
strategy. Additionally, the extra $L^{1}$-convergence in the theorems
is simply understood as a condition to also imply the optimality in
terms of the expected growth rate of the proposed strategies, as in
Corollary \ref{Corollary 1}.

\subsection{Market-knowledge-free learning algorithm}

In a stationary market, since the random optimal constant strategy
is invariant in time, it does not need to observe any side information
or know the possibly latent dependence structures of the process.
Although the optimal constant strategy cannot be identified beforehand,
a learning algorithm-based dynamic strategy could gradually approach
the growth rate of this constant strategy over time. Moreover, because
this type of algorithm inherits the advantages of a constant strategy,
it only requires learning from past assets' returns, which are always
available and not costly, but not from any side information of past
market events. It is worth noting that there are several algorithms
in the literature for this problem, with different assumptions and
complexities, which are not the focus of this paper, as surveyed in
\citet{Erven2020}. However, considering stationarity, the following
two learning algorithms are proposed for strategy development.\smallskip

\textbf{Cover\textquoteright s Universal Strategy}. The Universal
strategy, originally proposed in \citet{Cover1991,Cover1996,Cover2006,Ordentlich1996},
and extended with side information in \citet{Cover1996a}, aims to
construct a strategy that grows asymptotically at the same rate as
the best strategy within the restricted set of constant strategies,
for each finite sequence of realizations of assets' returns, from
the starting point to the last observations. This learning algorithm
is model-free and does not impose any assumptions on the market process.
Consequently, the growth rate of the Universal strategy will approach
that of the log-optimal strategy as time progresses, provided that
there exists an optimal constant strategy comparable to the log-optimal
one in the market. Moreover, it is important to note that Cover and
many subsequent works in the literature emphasize that this strategy
will almost surely attain the optimal growth rate only in an i.i.d.
market process of sole assets' returns, which is trivially a special
case of our more general theorem. Hence, our paper extends the optimality
of the Universal strategy to broader scenarios involving markets with
side information, far beyond the simple i.i.d. process.\smallskip

In detail, the construction of the Universal portfolio strategy is
as follows:
\[
\bar{b}_{1}\coloneqq\big(1/m,...,1/m\big)\text{ and }\bar{b}_{n}\coloneqq{\displaystyle \dfrac{{\displaystyle \int_{\mathcal{B}^{m}}{\displaystyle b}}S_{n-1}\big(b\big)\mu\big(b\big)\dif b}{{\displaystyle \int_{\mathcal{B}^{m}}S_{n-1}\big(b\big)\mu\big(b\big)\dif b}}},\,\forall n\geq2,
\]
where $\mu(\cdot)$ denotes the fixed uniform density. Essentially,
the algorithm distributes the current cumulative wealth according
to the previous weights of cumulative wealth of all constant strategies,
thereby learning only from the past portfolio returns of these component
strategies. Since the algorithm operates without additional assumptions
on assets' returns, the Universal portfolio can be constructed directly
using the variables $X_{n}$, rather than their normalizations. However,
the main drawback of this algorithm is its computational difficulty
in practice, particularly when the portfolio involves a large number
of invested assets (for example, $m\geq3$), which motivates the development
of a simpler algorithm with the same optimality, as described next.\smallskip 

\textbf{Empirical distribution-based log-optimal strategy}. This algorithm
is naturally developed from the argument in the proof of Theorem \ref{Theorem 2},
but with an extension to include the market with non-decaying impact
of information. In general, this learning algorithm is much simpler
than the Universal strategy, as it constructs the strategy by sequentially
solving for the log-optimal portfolios corresponding to the empirical
distribution of realizations of assets' returns. Despite the algorithm's
simplicity, such an empirical strategy, as proposed below, is comparable
with the Universal strategy in general markets, as long as there exists
an optimal constant strategy that is comparable to the log-optimal
one, as demonstrated in Theorem \ref{Theorem 3}.\smallskip 

Specifically, we introduce the dynamic strategy $\big(\bar{b}_{n}\big)$
constructed as follows:
\begin{equation}
\bar{b}_{1}\coloneqq\big(1/m,...,1/m\big)\text{ and }\mathbb{E}^{P_{n-1}}\big(\log\big<\bar{b}_{n},U\big>\big)\coloneqq\max_{b\in\mathcal{B}^{m}}\mathbb{E}^{P_{n-1}}\big(\log\big<b,U\big>\big),\,\forall n\geq2,\label{empirical log-strategy 1}
\end{equation}
where $\mathbb{E}^{P_{n}}(\cdot)$ denotes the expected value corresponding
to $U\sim P_{n}$, which is defined as:
\begin{equation}
P_{n}\left(A\right)\coloneqq\dfrac{1}{n}\sum_{i=1}^{n}\mathbb{I}_{U_{i}}\left(A\right),\,\forall A\subseteq\mathcal{U},\forall n,\label{empirical log-strategy 2}
\end{equation}
i.e., an empirical distribution of the sequence of distributions $\big\{ P_{n}\big\}_{n=1}^{\infty}$.
Noting that the empirical distributions are computed simply using
the normalized realizations of assets' returns.\smallskip 
\begin{rem*}
In the literature, a similar learning algorithm with a comparable
mechanism is proposed by \citet{Morvai1991}. However, this paper
only proposes a strategy that captures the growth rate of the log-optimal
strategy in a stationary and ergodic market without side information,
assuming integrability. Moreover, the strategy is stated to be optimal
only if the market is i.i.d., as concluded by the author based on
Cover's assessment. In comparison, Theorem \ref{Theorem 3} demonstrates
the optimality of the empirical distribution-based log-optimal strategy
in a more general market, without requiring stationarity, and provides
a simpler proof that does not rely on the integrability assumption.
\end{rem*}
\begin{lem}
Consider a sequence of positive real values $\big\{ a_{n}\big\}_{n=1}^{\infty}$.
If the following holds:\label{AM-GM lemma}
\[
\lim_{n\to\infty}\big(\sum_{i=1}^{n}a_{i}-n\big(\prod_{i=1}^{n}a_{i}\big)^{1/n}\big)=0,
\]
then $\max_{1\leq i,j\leq n}|a_{i}-a_{j}|\to0$ as $n\to\infty$.
\end{lem}
\begin{proof}
Let denote $\delta_{n}\coloneqq\big(\sqrt{a_{i_{n}}}-\sqrt{a_{j_{n}}}\big)^{2}\coloneqq\max_{1\leq i,j\leq n}\big(\sqrt{a_{i}}-\sqrt{a_{j}}\big)^{2}$
for all $n$, which are strictly positive if there exist two distinct
terms in the finite sequence $\big\{ a_{i}\big\}_{i=2}^{n}$. Then
we have:
\begin{align*}
\sum_{i=1}^{n}a_{i}-n\big(\prod_{i=1}^{n}a_{i}\big)^{1/n} & =\delta_{n}+\sum_{i\notin\left\{ i_{n},j_{n}\right\} }a_{i}+2\sqrt{a_{i_{n}}a_{j_{n}}}-n\big(\prod_{i=1}^{n}a_{i}\big)^{1/n}\\
 & \geq\delta_{n}+n\big(\prod_{i\notin\left\{ i_{n},j_{n}\right\} }a_{i}a_{i_{n}}a_{j_{n}}\big)^{1/n}-n\big(\prod_{i=1}^{n}a_{i}\big)^{1/n}=\delta_{n},
\end{align*}
due to the AM-GM inequality. Moreover, since $\big(\sqrt{a_{i}}-\sqrt{a_{j}}\big)^{2}\geq|a_{i}-a_{j}|$
for all $i,j\in\left\{ 1,...,n\right\} $, we consequently obtain
the following limit:
\begin{align*}
0=\lim_{n\to\infty}\big(\sum_{i=1}^{n}a_{i}-n\big(\prod_{i=1}^{n}a_{i}\big)^{1/n}\big) & \geq\lim_{n\to\infty}\max_{1\leq i,j\leq n}\big(\sqrt{a_{i}}-\sqrt{a_{j}}\big)^{2}\\
 & \geq\lim_{n\to\infty}\max_{1\leq i,j\leq n}|a_{i}-a_{j}|=0,
\end{align*}
which is the desired result and completes the proof.
\end{proof}
\begin{thm}
Assume that there exists an optimal constant strategy $\big(b^{\omega}\big)$,
which can be either random or non-random, in the market process $\big\{ X_{n},Y_{n}\big\}_{n=1}^{\infty}$
(not necessarily stationary), such that $\big(W_{n}\big(b^{\omega}\big)-W_{n}\big(b_{n}^{*}\big)\big)\to0$
almost surely. Then, the strategy $\big(\bar{b}_{n}\big)$, constructed
according to (\ref{empirical log-strategy 1}) and (\ref{empirical log-strategy 2}),
is also optimal in terms of asymptotic growth rate, as the log-optimal
strategy, satisfying:\label{Theorem 3}
\[
\lim_{n\to\infty}\big(W_{n}\big(\bar{b}_{n}\big)-W_{n}\big(b_{n}^{*}\big)\big)=\lim_{n\to\infty}\big(W_{n}\big(\bar{b}_{n}\big)-W_{n}\big(b^{\omega}\big)\big)=0,\text{ a.s},
\]
which are also in $L^{1}$ under the additional condition of no-trash
assets. As a consequence, the strategy $\big(\bar{b}_{n}\big)$ is
not only optimal in a stationary general market, but also in one with
non-decaying impact of past information on future market events, as
described in Proposition \ref{Proposition 1}.
\end{thm}
\begin{proof}
First, consider the case of a stationary market and provide a specific
proof for the theorem statement. Recall the steps in (\ref{Limit equality})
of the proof of Theorem \ref{Theorem 2}, we have:
\begin{align*}
\lim_{n\to\infty}\frac{1}{n}\sum_{i=1}^{n}\big(\log\left\langle b^{\omega},U_{i}\right\rangle -\log\left\langle b_{i}^{*},U_{i}\right\rangle \big) & =\lim_{n\to\infty}\frac{1}{n}\sum_{i=1}^{n}\big(\log\big<b^{*,n},U_{i}\big>-\log\left\langle b_{i}^{*},U_{i}\right\rangle \big)=0,
\end{align*}
where we define the constant strategy $\big(b^{*,n}\big)$ as the
portfolio $b^{*,n}$ that maximizes the following expectation with
respect to the empirical distribution $P_{n}$ of the normalized assets'
returns $U$:
\[
\mathbb{E}^{P_{n}}\left(\log\left\langle b^{*,n},U\right\rangle \right)=\max_{b\in\mathcal{B}^{m}}\mathbb{E}^{P_{n}}\left(\log\left\langle b,U\right\rangle \right),
\]
with the empirical distribution $P_{n}$ defined according to (\ref{empirical log-strategy 2})
for all $n$. To avoid confusion, it should be noted that the strategy
$\big(\bar{b}_{n}\big)$ is dynamic, while the strategies $\big(b^{*,n}\big)$
are constant and depend on $P_{n}$, which are denoted here for convenience
in later deductions.\smallskip 

Next, we deduce the weak convergence of the empirical distribution
sequence as $P_{n}\to\mathbb{P}\big(U_{1}|\mathcal{I}\big)$ almost
surely, due to Birkhoff's ergodic theorem. This results, almost surely,
in that $\big<\bar{b}_{n},U_{1}\big>\to\big<b^{*},U_{1}\big>$ for
$\mathbb{P}\big(U_{1}|\mathcal{I}\big)$-almost all $U_{1}$, with
$b^{*}$ maximizing $\mathbb{E}\big(\log\big<b,U_{1}\big>|\mathcal{I}\big)=\mathbb{E}\big(\log\big<b,U_{k}\big>|\mathcal{I}\big)$
for all $k>1$, by Lemma \ref{Lemma 4}. Thus, we also have $\big(\log\big<\bar{b}_{n},U_{k}\big>-\log\big<b^{*},U_{k}\big>\big)\to0$
for all $k$ almost surely, so $\big(\log\big<\bar{b}_{n},U_{n}\big>-\log\big<b^{*},U_{n}\big>\big)\to0$
also holds almost surely, by Lemma \ref{Lemma 5}. The needed equality
is then obtained as follows, due to Cesaro's mean theorem and Theorem
\ref{Theorem 2}:
\[
\lim_{n\to\infty}\frac{1}{n}\sum_{i=1}^{n}\big(\log\big<\bar{b}_{i},U_{i}\big>-\log\big<b^{*},U_{i}\big>\big)=\lim_{n\to\infty}\frac{1}{n}\sum_{i=1}^{n}\big(\log\big<\bar{b}_{i},U_{i}\big>-\log\big<b_{n}^{*},U_{i}\big>\big)=0,\text{ a.s}.
\]
Hence, under the additional condition of no-trash assets, the almost
sure convergence above is also in $L^{1}$ using Lebesgue's dominated
convergence theorem.\smallskip

\emph{(General case without the strict stationarity condition)}. If
there exists a random optimal constant strategy but the market is
not stationary, then the weak convergence of the empirical distribution
$P_{n}\to\mathbb{P}\big(U_{1}|\mathcal{I}\big)$ is not guaranteed,
so the following arguments of the aforementioned proof fail. Instead,
by utilizing the properties of the log-optimal portfolios, we can
argue that the optimal constant strategy should also be almost optimal
with respect to the empirical distribution $P_{n}$ as $n\to\infty$.
Indeed, this property can be established using Lemma \ref{AM-GM lemma}
through the following points.\smallskip 

\emph{Point 1}. Since there exists a random optimal constant strategy
$\big(b^{\omega}\big)$ by assumption, we have:
\begin{align*}
0=\lim_{n\to\infty}\Big(\frac{1}{n}\sum_{i=1}^{n}\log\big<b^{\omega},U_{i}\big>-\max_{b\in\mathcal{B}^{m}}\frac{1}{n}\sum_{i=1}^{n}\log\big<b,U_{i}\big>\Big) & =\lim_{n\to\infty}\frac{1}{n}\sum_{i=1}^{n}\Big(\log\frac{\big<b^{\omega},U_{i}\big>}{\big<b^{*,n},U_{i}\big>}\Big)\\
 & \leq\lim_{n\to\infty}\log\Big(\frac{1}{n}\sum_{i=1}^{n}\frac{\big<b^{\omega},U_{i}\big>}{\big<b^{*,n},U_{i}\big>}\Big)\\
 & =\lim_{n\to\infty}\log\mathbb{E}^{P_{n}}\Big(\frac{\big<b^{\omega},U\big>}{\big<b^{*,n},U\big>}\Big)\leq0,
\end{align*}
due to Jensen's inequality and Lemma \ref{Lemma 2}, taking into account
that the portfolios $\bar{b}_{n+1}=b^{*,n}$ are the maximizers for
$\max_{b\in\mathcal{B}^{m}}\mathbb{E}^{P_{n}}\big(\log\big<b,U\big>\big)$
corresponding to the empirical distributions $P^{n}$. As a result,
the obtained equality implies:
\begin{equation}
\lim_{n\to\infty}\frac{1}{n}\sum_{i=1}^{n}\frac{\big<b^{\omega},U_{i}\big>}{\big<b^{*,n},U_{i}\big>}=\lim_{n\to\infty}\Big(\prod_{i=1}^{n}\frac{\big<b^{\omega},U_{i}\big>}{\big<b^{*,n},U_{i}\big>}\Big)^{1/n}=1,\label{Equality of two limits}
\end{equation}
which immediately leads to the following consequence by Lemma \ref{AM-GM lemma}:
\begin{equation}
\lim_{n\to\infty}\max_{1\leq i,j\leq n}\Big|\frac{\big<b^{\omega},U_{i}\big>}{\big<b^{*,n},U_{i}\big>}-\frac{\big<b^{\omega},U_{j}\big>}{\big<b^{*,n},U_{j}\big>}\Big|=0.\label{Equality of all terms}
\end{equation}
This limit holds only if all terms of the sequence $\big\{\big<b^{\omega},U_{i}\big>/\big<b^{*,n},U_{i}\big>\big\}_{i=1}^{n}$
converge uniformly to the same value $1$ as $n\to\infty$. Indeed,
since the equality (\ref{Equality of all terms}) also implies that:
\[
\lim_{n\to\infty}\Big(\max_{1\leq i\leq n}\Big\{\frac{\big<b^{\omega},U_{i}\big>}{\big<b^{*,n},U_{i}\big>}\Big\}-\min_{1\leq i\leq n}\Big\{\frac{\big<b^{\omega},U_{i}\big>}{\big<b^{*,n},U_{i}\big>}\Big\}\Big)=0,
\]
thus, by using the convergence in (\ref{Equality of two limits}),
we obtain the needed result:
\[
1=\lim_{n\to\infty}\min_{1\leq i\leq n}\Big\{\frac{\big<b^{\omega},U_{i}\big>}{\big<b^{*,n},U_{i}\big>}\Big\}\leq\lim_{n\to\infty}\frac{1}{n}\sum_{i=1}^{n}\frac{\big<b^{\omega},U_{i}\big>}{\big<b^{*,n},U_{i}\big>}\leq\lim_{n\to\infty}\max_{1\leq i\leq n}\Big\{\frac{\big<b^{\omega},U_{i}\big>}{\big<b^{*,n},U_{i}\big>}\Big\}=1.
\]
Hence, $\big(\log\big<b^{\omega},U_{n}\big>-\log\big<b^{*,n},U_{n}\big>\big)\to0$
as $n\to\infty$, due to $\big(\log\big<b^{\omega},U_{i}\big>-\log\big<b^{*,n},U_{i}\big>\big)\to0$
uniformly for all $i\in\left\{ 1,...,n\right\} $ as $n\to\infty$.
In order to conclude the convergence of the strategy $\big(\bar{b}_{n}\big)$,
we need to establish the following additional property.\smallskip 

\emph{Point 2}. Noting that although the convergence of the sequence
of distributions $\big\{ P_{n}\big\}_{n=1}^{\infty}$ is not guaranteed,
we still have the following convergence for all bounded continuous
functions $f(\cdot)$:
\begin{align*}
\lim_{n\to\infty}\big(\mathbb{E}^{P_{n+1}}\big(f\big(U\big)\big)-\mathbb{E}^{P_{n}}\big(f\big(U\big)\big)\big) & =\lim_{n\to\infty}\Big(\frac{1}{n+1}\sum_{i=1}^{n+1}f\big(U_{i}\big)-\frac{1}{n}\sum_{i=1}^{n}f\big(U_{i}\big)\Big)\\
 & =\lim_{n\to\infty}\Big(\frac{1}{n+1}f\big(U_{n+1}\big)-\frac{1}{n\big(n+1\big)}\sum_{i=1}^{n}f\big(U_{i}\big)\Big)=0.
\end{align*}
Therefore, by Lemma \ref{Lemma 4}, the function $\max_{b\in\mathcal{B}^{m}}\mathbb{E}\big(\log\big<b,U\big>\big)$
is continuous over the space of distributions equipped with the weak
topology. This leads to the following:
\[
\lim_{n\to\infty}\big(\max_{b\in\mathcal{B}^{m}}\mathbb{E}^{P_{n+1}}\big(\log\big<b,U\big>\big)-\max_{b\in\mathcal{B}^{m}}\mathbb{E}^{P_{n}}\big(\log\big<b,U\big>\big)\big)=0.
\]
Moreover, Lemma \ref{Lemma 4} also implies that if the log-optimal
portfolios for $\max_{b\in\mathcal{B}^{m}}\mathbb{E}^{P_{n}}\big(\log\big<b,U\big>\big)$
are unique for all $n$, then $\big(b^{*,n}-b^{*,n-1}\big)\to0$ as
$n\to\infty$. Additionally, recalling that this lemma further asserts
that $\big<\hat{b},U\big>=\big<\bar{b},U\big>$ for any two maximizers
$\hat{b}$ and $\bar{b}$ for $\max_{b\in\mathcal{B}^{m}}\mathbb{E}^{P_{n}}\big(\log\big<b,U\big>\big)$
for all $P_{n}$, the following convergence always holds:
\[
\lim_{n\to\infty}\big(\log\big<b^{*,n},U_{n}\big>-\log\big<b^{*,n-1},U_{n}\big>\big)=0,
\]
taking into account that all the normalized assets' returns $U_{n}\in\mathcal{U}$
are bounded for all $n$, as stated in (\ref{U is bounded}) in the
proof of Lemma \ref{Lemma 1}.\smallskip 

\emph{(Definitive statement of the optimality)}. Finally, we establish
the convergence for the strategy $\big(\bar{b}_{n}\big)$, where $\bar{b}_{n}=b^{*,n-1}$
by construction of the strategy:
\begin{align*}
\lim_{n\to\infty}\big(\log\big<b^{\omega},U_{n}\big>-\log\big<\bar{b}_{n},U_{n}\big>\big)= & \lim_{n\to\infty}\big(\log\big<b^{\omega},U_{n}\big>-\log\big<b^{*,n},U_{n}\big>\big)\\
 & +\lim_{n\to\infty}\big(\log\big<b^{*,n},U_{n}\big>-\log\big<b^{*,n-1},U_{n}\big>\big)=0.
\end{align*}
Then, invoking the Cesaro's mean theorem again, we obtain $\big(W_{n}\big(\bar{b}_{n}\big)-W_{n}\big(b_{n}^{*}\big)\big)\to0$
as:
\[
\lim_{n\to\infty}\frac{1}{n}\sum_{i=1}^{n}\big(\log\big<b^{\omega},U_{i}\big>-\log\big<\bar{b}_{n},U_{i}\big>\big)=\lim_{n\to\infty}\frac{1}{n}\sum_{i=1}^{n}\big(\log\big<b_{i}^{*},U_{i}\big>-\log\big<\bar{b}_{n},U_{i}\big>\big)=0.
\]
This results in $L^{1}$-convergence by the Lebesgues' dominated convergence
theorem under the additional condition of no-trash assets. Thus, the
proof is completed.
\end{proof}

\subsection{Information-observation-based learning algorithm}

After establishing the asymptotic optimality of the log-optimal strategy
in \citet{Algoet1988}, a problem arose concerning the construction
of a strategy that asymptotically approaches the same finite growth
rate as the log-optimal strategy, under a stationary and ergodic market
with an unknown infinite-dimensional distribution. This problem was
posed in Algoet's later work \citet{Algoet1992}, where it was referred
to as the desired universality of a so-called optimal strategy. In
addition to formulating the problem, Algoet also sketched a learning
algorithm to construct a strategy that could achieve the target limiting
growth rate, based on the learning algorithms developed in the thesis
of \citet{Bailey1976} and later by \citet{Ornstein1978}. However,
this algorithm was known for its practical difficulties until the
work by \citet{Gyorfi2006}, which made the original idea more practical.
Since then, variants of Gyorfi's learning algorithm (a collective
term for Gyorfi and his co-authors) have been regarded as the primary
class of methods that fully implement Algoet's idea in contemporary
literature. In this section, we extend the investigation of the properties
of this method and enhance its practicality, especially when the market
is stationary but not ergodic, with side information and the guaranteed
existence of an optimal constant strategy.\smallskip 

\textbf{Anatomy of Algoet's and Gyorfi's learning algorithms}. As
mentioned in the previous section, since the true ergodic mode of
a stationary process and the associated conditional distribution cannot
be exactly identified from any finite observation of past events,
a universal estimation scheme is formulated by learning the realizations
of the process to approximate the true conditional distribution given
the infinite past. In the paper \citet{Algoet1992}, the author generalizes
the universal estimation scheme for finite-valued processes developed
by \citet{Bailey1976} and \citet{Ornstein1978} to the Polish space,
which guarantees the asymptotic optimal growth rate for a log-optimal
strategy derived using the estimated conditional distribution. However,
due to the complexity of its estimation, Algoet proposes another simpler
algorithm that generates log-optimal portfolios based on conditional
distributions estimated using finite-order Markov approximation functions,
then combines them using the \emph{bookkeeping technique}, commonly
referred to as the \emph{exponentially weighted average} (as in \citet{CesaBianchi2006})
in the present literature. The main practical issue with this algorithm
is that the Markov approximations only capture similar events belonging
to the same atoms in a $\sigma$-field, so the learning process often
results in void sets. To address this problem, Gyorfi\textquoteright s
algorithm uses the same mechanism, but the combination procedure is
conducted with finite-order Markov approximations that capture various
distances between two events, rather than requiring absolute equality.\smallskip 

\textbf{Finite-order Markov approximation and the estimated log-optimal
portfolio}. Despite the improved practicality, it should be stressed
that there are two inherent challenges in both the original Algoet's
algorithm and the later Gyorfi's algorithm: first, they require complete
observation of all available market side information, even though
parts of it may be unobservable, which \citet{Algoet1992} explicitly
mentions as a conjectured result for a possible extension of the algorithm
to markets with side information \footnote{See Remark 2 in Section 4.4 for a discussion of the sufficiency of
using side information for conditional distribution estimation, and
Remark 2 in Section 5.2 for the conjectured result for a stationary
and ergodic market.}; second, the Markov approximation must be taken to an infinite order,
rendering the available data insufficient for the learning procedure
to guarantee convergence of the estimation. Hence, although the optimality
of the estimated log-optimal strategy is theoretically guaranteed,
it seems unrealistic in practice. However, based on the results obtained
in the previous section, we can expect some possible simplifications
to the algorithm.\smallskip 

Without loss of generality, by assuming that all market side information
is observable, consider the following estimated log-optimal strategy
$\big(b_{n}^{(h,l)}\big(\theta\big)\big)$ based on a Markov approximation
of order $h$ with respect to normalized assets' returns, and a parameter
vector $\theta\in\mathbb{R}_{++}^{m\times(h-1)}\times\mathbb{R}^{k\times h}$:
\begin{equation}
b_{n}^{(h,l)}\big(\theta\big)\coloneqq\operatorname*{argmax}_{b\in\mathcal{B}^{m}}\frac{1}{|\mathcal{D}_{n}^{(h,l)}\big(\theta\big)|}\sum_{X_{i}\in\mathcal{D}_{n}^{(h,l)}(\theta)}\log\big<b,U_{i}\big>,\,\forall n\geq h+2,\label{kernel 1}
\end{equation}
which implies that $b_{n}^{(h,l)}\big(\theta\big)$ maximizes $\mathbb{E}^{P(\cdot|\mathcal{D}_{n}^{(h,l)}(\theta))}\big(\log\big<b,u\big(X\big)\big>\big)$
for all $n>h$, where:
\[
X\sim P\big(\cdot|\mathcal{D}_{n}^{(h,l)}\big(\theta\big)\big),\text{ and }P\big(A|\mathcal{D}_{n}^{(h,l)}\big(\theta\big)\big)\coloneqq\frac{1}{|\mathcal{D}_{n}^{(h,l)}\big(\theta\big)|}\sum_{X_{i}\in\mathcal{D}_{n}^{(h,l)}(\theta)}\mathbb{I}_{X_{i}}\big(A\big),\,\forall A\subset\mathbb{R}_{++}^{m},
\]
and $h,l\in\mathbb{N}_{+}$, while $c\in\mathbb{R}_{++}$ is a fixed
small value. Also,
\begin{equation}
\mathcal{D}_{n}^{(h,l)}\big(\theta\big)\coloneqq\big\{ X_{i}:\,\left\Vert \big(X_{i-h}^{i-1},Y_{i-h}^{i}\big)-\theta\right\Vert \leq c/l,\text{ }h+1\leq i\leq n-1\big\},\,\forall n\geq h+2,\label{kernel 2}
\end{equation}
with $b_{n}\coloneqq\big(1/m,...,1/m\big)$ and $P\big(\cdot|\mathcal{D}_{n}^{(h,l)}\big(\theta\big)\big)\coloneqq\delta_{(1,..,1)}(\cdot)$
(Dirac measure) for an empty $\mathcal{D}_{n}^{(h,l)}\big(\theta\big)$
by convention.\smallskip 

Since the market process is assumed to be stationary, an investor
could start investing after $h$ periods, which is equivalent to moving
the origin, and all estimated log-optimal portfolios will be formed
from the first period onward. The sets $\mathcal{D}_{n}^{\left(h,l\right)}\big(\theta\big)$
represent collections of asset returns having finite sequences of
past asset returns with similar patterns, measured by a norm (such
as the Frobenius norm). Consider $\theta=\big(X_{n-h}^{n-1},Y_{n-h}^{n}\big)$
as an instance. The norm of similarity between two events is computed
using matrices as follows:
\[
\left\Vert \big(X_{i-h}^{i-1},Y_{i-h}^{i}\big)-\big(X_{n-h}^{n-1},Y_{n-h}^{n}\big)\right\Vert =\left\Vert \left[\begin{array}{cccc}
Y_{i-h} & \cdots & Y_{i-1} & Y^{i}\\
X_{i-h} & \cdots & X_{i-1} & 0
\end{array}\right]-\left[\begin{array}{cccc}
Y_{n-h} & \cdots & Y_{n-1} & Y^{n}\\
X_{n-h} & \cdots & X_{n-1} & 0
\end{array}\right]\right\Vert .
\]
Additionally, the value $c$ is fixed, and $c/l\to0$ as $l\to\infty$,
which partitions the $\sigma$-field $\sigma\big(X_{n-h}^{n-1},Y_{n-h}^{n}\big)$
into coarser approximation $\sigma$-fields for each level $c/l$.
As $l\to\infty$, the approximation becomes finer, and all similar
events must belong to the same atoms of the $\sigma$-field $\sigma\big(X_{n-h}^{n-1},Y_{n-h}^{n}\big)$.
Furthermore, it is worth noting that there are various methods to
measure the similarity between events, as discussed in \citet{Gyorfi2006}.
However, for representational purposes, we focus on the partition
in (\ref{kernel 2}), which is commonly referred to as the \emph{kernel-based
criterion} in the literature, due to its empirically outstanding performance
and implementability.\smallskip 

\textbf{Strategy constructed by mixing kernels with fixed order}.
The mechanism of measuring similarity between events across several
extents $c/l$ allows the partitions to collect more elements during
the scanning process of the past, thereby avoiding the empty sets
$\mathcal{D}_{n}^{(h,l)}\big(\theta\big)$. This is especially important
since the amount of available data is limited when the time horizon
is small. Hence, similar to Gyorfi's argument, we construct the following
combinatorial strategy by exponentially averaging the weighted strategies
$\big(b_{n}^{(h,l)}\big(\theta_{n}\big)\big)$ corresponding to all
values of $l$ ranging over $\mathbb{N}_{+}$, where $\theta_{n}\coloneqq\big(X_{n-h}^{n-1},Y_{n-h}^{n}\big)$,
as:
\begin{equation}
b_{n}^{(h)}\coloneqq\left(1/m,...,1/m\right),\forall n<h+2\text{, and }b_{n}^{(h)}\coloneqq{\displaystyle \frac{\sum_{l}b_{n}^{(h,l)}\big(\theta_{n}\big)S_{n-1}\big(b_{n-1}^{(h,l)}\big(\theta_{n-1}\big)\big)}{\sum_{l}S_{n-1}\big(b_{n-1}^{(h,l)}\big(\theta_{n-1}\big)\big)}},\,\forall n\geq h+2\label{kernel combine}
\end{equation}
Noting that the combination in (\ref{kernel combine}) fixes an $h$,
which differs from the approach in Algoet's and Gyorfi's algorithms,
which combine all orders of $h$ over $\mathbb{N}_{+}$. This distinction
arises because the algorithm cannot learn from past data when $h\geq n$,
and thus, the theoretically guaranteed optimality is not attainable
in practice. This limitation leads to the restriction of choosing
order $h$ in the experiments in \citet{Gyorfi2006}\footnote{In the remark on the ``Validity of assumptions'' and the second
paragraph of Section 4.1 of the cited article, the authors note that
certain component strategies with smaller order $h$ perform better
than others. Therefore, $h$ can be restricted, as the market may
be a low-order Markov process. Moreover, such low-order Markov component
strategies effectively exploit the hidden dependence structure of
the market, which is difficult to reveal. Clearly, this is the primary
motivation of our paper: to address the challenge of unobservable
latent market dependence.}.\smallskip 

The asymptotic growth rate and optimality of the strategy $\big(b_{n}^{(h)}\big)$
are established in Theorem \ref{Theorem 4} with an arbitrarily fixed
order $h$. The theorem asserts that the choice of order does not
affect the optimality of the strategy $\big(b_{n}^{(h)}\big)$, and
therefore, there is no need to mix all orders $h$ over $\mathbb{N}_{+}$.
Moreover, this result holds even if the algorithm is restricted to
learning only from past asset returns. In this context, the limiting
growth rate of the strategy $\big(b_{n}^{(h)}\big)$ is given by $\mathbb{E}\big(\max_{b\in\mathcal{B}^{m}}\mathbb{E}\big(\log\big<b,U_{1}\big>|\sigma\big(X_{1-h}^{0}\big)\big)|\mathcal{I}\big)$
under a stationary but not ergodic market process. This statement
is easily justified following the final derivation for the optimality
of the strategy $\big(b_{n}^{(h)}\big)$ in the proof of Theorem \ref{Theorem 4}.
Based on this property, it is safe to remove the side information
from the learning procedure of the algorithm.\smallskip 

\textbf{Notes on mixing kernels and markets with decaying impact of
information}. In the case where the stationary market is either an
i.i.d. or a finite-order Markov process, with $h^{*}\geq0$ denoting
the true order of the market's memory, the strategy $\big(b_{n}^{(h)}\big)$
must be constructed with a chosen $h\geq h^{*}$ to activate the Markov
property during the learning procedure, thereby enabling Theorem \ref{Theorem 4}
with stationarity. In this market scenario, the combinatorial strategy
$\big(b_{n}^{(h^{*})}\big)$ with true order $h^{*}$, learns the
log-optimal strategy $\big(b_{n}^{*}\big)$, which is a sequence of
functionally generated portfolios $b_{n}^{*}$ defined by the function
$\operatorname*{argmax}_{b\in\mathcal{B}^{m}}\mathbb{E}\big(\log\big<b,U_{1}\big>|\sigma\big(X_{1-h}^{0},Y_{1-h}^{1}\big)\big)$
for all $n$. Let us recall that this result generalizes the work
of \citet{Cuchiero2019} for the same type of market, as mentioned
in the last remark of the previous section. However, the strategy
$\big(b_{n}^{(h^{*})}\big)$ is a mixing of estimated log-optimal
strategies based on $h^{*}$-order Markov approximation functions
with different kernel widths, rather than mixing all functions over
their entire space, which seems impractical.\smallskip 

However, since the market's dependence structure is latent, the choice
of order $h$ for the strategy $\big(b_{n}^{(h^{*})}\big)$ should
be large enough to cover all possibilities of the unknown true order
$h^{*}$ of the market's memory length. In practice, because the sets
$\mathcal{D}_{n}^{\left(h,l\right)}\big(\theta\big)$ may become almost
void if $h$ is large, we should mix all orders within a threshold
$H$ to construct the strategy $\big(\bar{b}_{n}\big)$ as follows:
\[
\bar{b}_{n}\coloneqq\left(1/m,...,1/m\right),\forall n<h+2\text{, and }\bar{b}_{n}\coloneqq{\displaystyle \frac{\sum_{h\leq H}b_{n}^{(h)}S_{n-1}\big(b_{n-1}^{(h)}\big)}{\sum_{h\leq H}S_{n-1}\big(b_{n-1}^{(h)}\big)}},\,\forall n\geq h+2.
\]
Then, due to the known property of the weighted exponential average
(as demonstrated in (\ref{WEA inequality}) of the proof of Theorem
\ref{Theorem 4}), the growth rate of the combinatorial strategy $\big(\bar{b}_{n}\big)$
in terms of normalized assets' returns will converge to that of the
component strategy $\big(b_{n}^{(h^{*})}\big)$ with true order $h^{*}$.
Finally, it should be noted that a critical weakness of the learning
algorithm for the strategy $\big(b_{n}^{(h)}\big)$ is that it requires
$l\to\infty$, which makes the subroutine's computation very costly.
Therefore, the market-knowledge-free learning algorithms described
in the previous section, which capture the growth rate of the random
optimal constant strategy, are generally still preferable.\smallskip 
\begin{rem*}
The proof of Theorem \ref{Theorem 4} follows a substantial portion
of the arguments in the proof of \citet{Gyorfi2006}, with a generalization
to the case of a stationary but not ergodic market with side information,
capturing a broader context where the optimal limiting growth rate
of a strategy may be finitely random or not well-defined. Additionally,
the theorem establishes the limiting growth rate for the strategy
$\big(b_{n}^{(h)}\big)$ with a fixed order $h$, in terms of normalized
assets' returns, rather than mixing all orders $h$.
\end{rem*}
\begin{thm}
If the market process $\big\{ X_{n},Y_{n}\big\}_{n=1}^{\infty}$ is
stationary (not necessarily ergodic) and free of trash assets, then
the strategy $\big(b_{n}^{(h)}\big)$, defined according to (\ref{kernel 1}),
(\ref{kernel 2}) and (\ref{kernel combine}), for any fixed $h\geq0$,
satisfies the following:\label{Theorem 4}
\[
\lim_{n\to\infty}\dfrac{1}{n}\sum_{i=1}^{n}\log\big<b_{i}^{(h)},U_{i}\big>=\mathbb{E}\big(\log\big<b_{h+1}^{*},U_{1}\big>|\mathcal{I}\big),\text{ \text{ a.s and in }\ensuremath{L^{1}}},
\]
where $b_{h+1}^{*}$ is the $\mathcal{F}_{h+1}^{X}$-measurable log-optimal
portfolio of the log-optimal strategy $\big(b_{n}^{*}\big)$. Therefore,
the strategy $\big(b_{n}^{(h)}\big)$ is also optimal in terms of
asymptotic growth rate, as:
\[
\lim_{n\to\infty}\big(W_{n}\big(b_{n}^{(h)}\big)-W_{n}\big(b^{*}\big)\big)=\lim_{n\to\infty}\big(W_{n}\big(b_{n}^{(h)}\big)-W_{n}\big(b_{n}^{*}\big)\big)=0,\text{ a.s and in }L^{1}.
\]
\end{thm}
\begin{proof}
For the sake of readability, we divide the proof into smaller parts
below.\smallskip 

\emph{(Letting time $n$ increase to infinity)}. Firstly, recall that
by stationarity, we have:
\[
\sigma\big(X_{n-h}^{0},Y_{1-h}^{1}\big)=T^{-(n-1)}\sigma\big(X_{n-h}^{n-1},Y_{n-h}^{n}\big),\,\forall n.
\]
Then, for any $\theta$, we show that the sequence of distributions
$\big\{ P\big(\cdot|\mathcal{D}_{n}^{(h,l)}\big(\theta\big)\big)\big\}_{n=1}^{\infty}$
converges weakly almost surely to the following distributions as $n\to\infty$:
\begin{equation}
\begin{cases}
\mathbb{P}\big(X_{1}|\left\Vert \big(X_{1-h}^{0},Y_{1-h}^{1}\big)-\theta\right\Vert \leq c/l,\mathcal{I}\big) & \text{if }\mathbb{P}\big(||\big(X_{1-h}^{0},Y_{1-h}^{1}\big)-\theta||\leq c/l\big)>0,\\
\delta_{(1,..,1)} & \text{if }\mathbb{P}\big(||\big(X_{1-h}^{0},Y_{1-h}^{1}\big)-\theta||\leq c/l\big)=0.
\end{cases}\label{let n to infty}
\end{equation}
Indeed, consider any bounded continuous function $f(\cdot)$, we obtain:
\begin{align*}
\lim_{n\to\infty}\mathbb{E}^{P(\cdot|\mathcal{D}_{n}^{(h,l)}(\theta))}f\big(X\big) & =\lim_{n\to\infty}\frac{|\text{ }h-n|}{|\text{ }h-n||\mathcal{D}_{n}^{(h,l)}\big(\theta\big)|}\sum_{X_{i}\in\mathcal{D}_{n}^{(h,l)}(\theta)}f\big(X_{i}\big)\\
 & =\frac{\mathbb{E}\big(f\big(X_{1}\big)\mathbb{I}_{||(X_{1-h}^{0},Y_{1-h}^{1})-\theta||\leq c/l}|\mathcal{I}\big)}{\mathbb{E}\big(\mathbb{I}_{||(X_{1-h}^{0},Y_{1-h}^{1})-\theta||\leq c/l}|\mathcal{I}\big)}\\
 & =\mathbb{E}\big(f\big(X_{1}\big)|\left\Vert \big(X_{1-h}^{0},Y_{1-h}^{1}\big)-\theta\right\Vert \leq c/l,\mathcal{I}\big),\text{ a.s,}
\end{align*}
in the case where $\mathbb{P}\big(||\big(X_{1-h}^{0},Y_{1-h}^{1}\big)-\theta||\leq c/l\big)>0$.
Meanwhile, if $\mathbb{P}\big(||\big(X_{1-h}^{0},Y_{1-h}^{1}\big)-s||\leq c/l\big)=0$,
the sequence $\big\{ P\big(\cdot|\mathcal{D}_{n}^{(h,l)}\big(\theta\big)\big)\big\}_{n=1}^{\infty}$
is identical to the Dirac delta distribution $\delta_{(1,..,1)}$
almost surely, which implies $\mathbb{E}^{P(\cdot|\mathcal{D}_{n}^{(h,l)}(\theta))}f\big(X\big)=f\big(1,...,1\big)$
for all $n$ almost surely.\smallskip 

\emph{(Letting extent $l$ increase to infinity)}. As $l\to\infty$,
the kernel width $c/l\to0$, so we expect the approximated conditional
distribution to converge to the true one as follows:
\begin{equation}
\lim_{n\to\infty}\mathbb{P}\big(X_{1}|\left\Vert \big(X_{1-h}^{0},Y_{1-h}^{1}\big)-\theta\right\Vert \leq c/l\big)=\mathbb{P}\big(X_{1}|\big(X_{1-h}^{0},Y_{1-h}^{1}\big)=\theta\big).\label{let l to infty}
\end{equation}
This desired convergence is demonstrated as follows, for any $A\subset\mathbb{R}_{++}^{m}$
and $\theta$:
\begin{align*}
\lim_{l\to\infty}\mathbb{P}\big(X_{1}\in A|\left\Vert \big(X_{1-h}^{0},Y_{1-h}^{1}\big)-\theta\right\Vert \leq c/l\big) & =\lim_{l\to\infty}\frac{\mathbb{P}\big(X_{1}\in A,\left\Vert \big(X_{1-h}^{0},Y_{1-h}^{1}\big)-\theta\right\Vert \leq c/l\big)}{\mathbb{P}\big(\left\Vert \big(X_{1-h}^{0},Y_{1-h}^{1}\big)-\theta\right\Vert \leq c/l\big)}\\
 & =\lim_{l\to\infty}\frac{\mathbb{E}\big(\mathbb{P}\big(X_{1}\in A|X_{1-h}^{0},Y_{1-h}^{1}\big)\mathbb{I}_{||(X_{1-h}^{0},Y_{1-h}^{1})-\theta||\leq c/l}\big)}{\mathbb{E}\big(\mathbb{I}_{||(X_{1-h}^{0},Y_{1-h}^{1})-\theta||\leq c/l}\big)}\\
 & =\mathbb{P}\big(X_{1}\in A|\big(X_{1-h}^{0},Y_{1-h}^{1}\big)=\theta\big),
\end{align*}
for $\mathbb{P}$-almost all $\theta\in\left\{ \big(X_{1-h}^{0},Y_{1-h}^{1}\big)\big(\Omega\big)\right\} $
by Lebesgue's density theorem.\smallskip 

\emph{(Deriving the first needed inequality)}. Let the portfolio $b_{\infty}^{(h,l)}\big(\theta\big)$
denote a maximizer for:
\[
\max_{b\in\mathcal{B}^{m}}\mathbb{E}\big(\log\big<b\big(\theta\big),U_{1}\big>|\left\Vert \big(X_{1-h}^{0},Y_{1-h}^{1}\big)-\theta\right\Vert \leq c/l,\mathcal{I}\big),
\]
in the case that $\mathbb{P}\big(||\big(X_{1-h}^{0},Y_{1-h}^{1}\big)-\theta||\leq c/l\big)>0$.
Otherwise, if $\mathbb{P}\big(||\big(X_{1-h}^{0},Y_{1-h}^{1}\big)-\theta||\leq c/l\big)>0$,
let $b_{\infty}^{(h,l)}\big(\theta\big)\coloneqq\big(1,...,1\big)$.\smallskip 

Then, fixing $l$, due to the convergence of distributions in (\ref{let n to infty})
as $n\to\infty$, we obtain almost surely the convergence $\log\big<b_{n}^{(h,l)}\big(\theta\big),U_{1}\big>\to\log\big<b_{\infty}^{(h,l)}\big(\theta\big),U_{1}\big>$
using Lemma (\ref{Lemma 4}). Hence, by applying Breiman's generalized
ergodic theorem, we get:
\begin{align}
\lim_{n\to\infty}\frac{1}{n}\sum_{i=1}^{n}\log\big<b_{i}^{(h,l)}\big(\theta_{i}\big),U_{i}\big> & =\lim_{n\to\infty}\frac{1}{n}\sum_{i=1}^{n}\log\big<b_{i}^{(h,l)}\big(X_{i-h}^{i-1},Y_{i-h}^{i}\big),U_{i}\big>\nonumber \\
 & =\lim_{n\to\infty}\frac{1}{n}\sum_{i=1}^{n}\log\big<b_{\infty}^{(h,l)}\big(T^{(i-1)}\big(X_{1-h}^{0},Y_{1-h}^{1}\big)\big),T^{(i-1)}U_{1}\big>\nonumber \\
 & =\mathbb{E}\big(\log\big<b_{\infty}^{(h,l)}\big(X_{1-h}^{0},Y_{1-h}^{1}\big),U_{1}\big>|\mathcal{I}\big),\text{ a.s and in }L^{1},\label{First needed inequality Theorem 3}
\end{align}
given that $\theta_{n}\coloneqq\big(X_{n-h}^{n-1},Y_{n-h}^{n}\big)$
as defined for the strategy $\big(b_{n}^{(h)}\big)$ in (\ref{kernel combine}),
and that the required $L^{1}$-domination is satisfied for the sequence
$\big\{\log\big<b_{n}^{(h,l)},U_{n}\big>\big\}_{n=1}^{\infty}$ due
to (\ref{uniformly integrable}), under the assumption of no-trash
assets.\smallskip 

Consequently, by letting $l\to\infty$, we obtain the convergence
of distributions in (\ref{let l to infty}) for $\theta=\big(X_{1-h}^{0},Y_{1-h}^{1}\big)$.
Hence, by Lemma \ref{Lemma 4} again, we have the following convergence:
\[
\lim_{n\to\infty}\log\big<b_{\infty}^{(h,l)}\big(X_{1-h}^{0},Y_{1-h}^{1}\big),U_{1}\big>=\log\big<b_{\infty}^{(h,\infty)}\big(X_{1-h}^{0},Y_{1-h}^{1}\big),U_{1}\big>,\text{ a.s,}
\]
where $b_{\infty}^{(h,\infty)}\big(X_{1-h}^{0},Y_{1-h}^{1}\big)$
denotes a maximizer for $\max_{b\in\mathcal{B}^{m}}\mathbb{E}\big(\log\big<b,U_{1}\big>|\sigma\big(X_{1-h}^{0},Y_{1-h}^{1}\big)\big)$.
Noting that this limit also holds for $\mathbb{P}\big(U_{1}|\mathcal{I}\big)$-almost
all $U_{1}$, almost surely. Next, taking the limit as $l\to\infty$
of both sides of (\ref{First needed inequality Theorem 3}), we obtain:
\begin{align*}
\lim_{l\to\infty}\lim_{n\to\infty}\frac{1}{n}\sum_{i=1}^{n}\log\big<b_{i}^{(h,l)}\big(\theta_{i}\big),U_{i}\big> & =\lim_{l\to\infty}\mathbb{E}\big(\log\big<b_{\infty}^{(h,l)}\big(X_{1-h}^{0},Y_{1-h}^{1}\big),U_{1}\big>|\mathcal{I}\big)\\
 & =\mathbb{E}\big(\log\big<b_{\infty}^{(h,\infty)}\big(X_{1-h}^{0},Y_{1-h}^{1}\big),U_{1}\big>|\mathcal{I}\big)\\
 & =\mathbb{E}\big(\log\big<b_{h+1}^{*},U_{1}\big>|\mathcal{I}\big),\text{ a.s,}
\end{align*}
where the last equality follows from the fact that $b_{\infty}^{(h,\infty)}\big(X_{1-h}^{0},Y_{1-h}^{1}\big)\big(T^{h}\omega\big)=b_{h+1}^{*}\big(\omega\big)$,
which is the measurable log-optimal portfolio corresponding to the
$\sigma$-field\textbf{ $\mathcal{F}_{h+1}^{X}=T^{h}\sigma\big(X_{1-h}^{0},Y_{1-h}^{1}\big)$}
due to the stationarity of the market process. Then, by telescoping,
we obtain the following:
\begin{align}
\lim_{n\to\infty}\frac{1}{n}\sum_{i=1}^{n}\log\big<b_{n}^{(h)},U_{i}\big> & =\lim_{n\to\infty}\frac{1}{n}\log\Big(\frac{1}{l}\sum_{l}\prod_{i=1}^{n}\big<b_{i}^{(h,l)},U_{i}\big>\Big)\nonumber \\
 & \geq\lim_{n\to\infty}\sup_{l}\frac{1}{n}\log\Big(\frac{1}{l}\prod_{i=1}^{n}\big<b_{i}^{(h,l)},U_{i}\big>\Big)\nonumber \\
 & \geq\sup_{l}\lim_{n\to\infty}\frac{1}{n}\Big(\log\frac{1}{l}+\sum_{i=1}^{n}\log\big<b_{n}^{(h)},U_{i}\big>\Big)\nonumber \\
 & \geq\lim_{l\to\infty}\lim_{n\to\infty}\frac{1}{n}\sum_{i=1}^{n}\log\big<b_{n}^{(h)},U_{i}\big>\label{WEA inequality}\\
 & =\mathbb{E}\big(\log\big<b_{h+1}^{*},U_{1}\big>|\mathcal{I}\big),\text{ a.s,}\nonumber 
\end{align}
which is the first needed inequality for the statement of the limit.\smallskip 

\emph{(Deriving the remaining inequality)}. Let a strategy $\big(b_{n}^{*(h)}\big)$
be defined such that at each time $n$, it looks back at the past
$\big(X_{n-h}^{n-1},Y_{n-h}^{n}\big)$ and selects the corresponding
optimal portfolio, i.e.,
\begin{align*}
\mathbb{E}\big(\log\big<b_{n}^{*(h)},U_{n}\big>|\sigma\big(X_{n-h}^{n-1},Y_{n-h}^{n}\big)\big) & =\max_{b\in\mathcal{B}^{m}}\mathbb{E}\big(\log\big<b,U_{n}\big>|\sigma\big(X_{n-h}^{n-1},Y_{n-h}^{n}\big)\big),\text{ }\forall n.
\end{align*}
This can be chosen from the origin $n=1$ as the market is a two-sided
infinite process. Clearly, all such optimal portfolios $b_{n}^{*(h)}$
could be represented by a single measurable log-optimal portfolio
$b_{h+1}^{*}\in\mathcal{F}_{h+1}^{X}$. Hence, applying Birkhoff's
ergodic theorem yields the following result:
\[
\lim_{n\to\infty}\frac{1}{n}\sum_{i=1}^{n}\log\big<b_{n}^{*(h)},U_{i}\big>=\lim_{n\to\infty}\frac{1}{n}\sum_{i=1}^{n}\log\big<b_{h+1}^{*},U_{i}\big>=\mathbb{E}\big(\log\big<b_{h+1}^{*},U_{1}\big>|\mathcal{I}\big)\text{, a.s.}
\]
Moreover, since all the combinatorial portfolios $b_{n}^{(h)}$ also
depend on past information $\big(X_{n-h}^{n-1},Y_{n-h}^{n}\big)$,
they are suboptimal to the portfolios $b_{n}^{*(h)}$. We can derive
the following inequality using the Kuhn-Tucker condition for log-optimality,
similar to the proof of Lemma \ref{Lemma 2}:
\[
\lim_{n\to\infty}\frac{1}{n}\sum_{i=1}^{n}\log\big<b_{n}^{(h)},U_{i}\big>\leq\lim_{n\to\infty}\frac{1}{n}\sum_{i=1}^{n}\log\big<b_{n}^{*(h)},U_{i}\big>=\mathbb{E}\big(\log\big<b_{h+1}^{*},U_{1}\big>|\mathcal{I}\big)\text{, a.s,}
\]
which finally constitutes the needed limit. Furthermore, the $L^{1}$-convergence
follows immediately from the application of Lebesgue's dominated convergence
theorem.\smallskip 

\emph{(Definitive statement of the optimality)}. Using the same argument
as before, and noting that any constant strategy based on a single
portfolio is measurable with respect to any $\mathcal{F}_{h+1}^{X}$
for all orders $h$, let us consider the optimal constant strategy
$\big(b^{*}\big)$ and the log-optimal strategy $\big(b_{n}^{*}\big)$
in a stationary market. We then obtain the following inequalities
due to Lemma \ref{Lemma 2}: 
\[
\lim_{n\to\infty}\frac{1}{n}\sum_{i=1}^{n}\log\big<b^{*},U_{i}\big>\leq\lim_{n\to\infty}\frac{1}{n}\sum_{i=1}^{n}\log\big<b_{n}^{(h)},U_{i}\big>\leq\lim_{n\to\infty}\frac{1}{n}\sum_{i=1}^{n}\log\big<b_{n}^{*},U_{i}\big>,\text{ a.s.}
\]
Since the terms associated with the strategies $\big(b^{*}\big)$
and $\big(b_{n}^{*}\big)$ are equal by Theorem \ref{Theorem 2},
we have:
\[
\lim_{n\to\infty}\big(W_{n}\big(b_{n}^{(h)}\big)-W_{n}\big(b^{*}\big)\big)=\lim_{n\to\infty}\big(W_{n}\big(b_{n}^{(h)}\big)-W_{n}\big(b_{n}^{*}\big)\big)=0,\text{ a.s},
\]
which establishes the optimality (also in $L^{1}$) of the strategy
$\big(b_{n}^{(h)}\big)$ and completes the proof.
\end{proof}

\section{Summary and concluding remarks}

In this paper, we investigate the problem of sequential portfolio
decision-making in a market with partially observable side information
and a latent dependence structure. The established results demonstrate
the existence of a random optimal constant strategy, which is time-invariant
and does not require knowledge of the dependence structure or observation
of additional market information. This strategy achieves an asymptotic
growth rate as fast as the log-optimal strategy, which dynamically
determines portfolios based on perfect market information over time.
The reason for this phenomenon lies in the stationarity of the market
process, which diminishes the advantages of fully utilizing market
information for dynamic strategies over time, leading to an equilibrium
state. Furthermore, if the market is also ergodic, a non-random optimal
constant strategy exists, regardless of the market's evolutionary
possibilities. These findings question conventional perspectives in
the fields of information theory, learning theory, and finance, where
the prevailing belief is that an optimal strategy should be dynamic
and utilize perfect information, rather than being a simple, time-invariant
one. The traditional view maintains that a constant strategy is optimal
only when the process of sole assets' returns is i.i.d.\smallskip

With the established equality theorems for the limiting growth rate
between the log-optimal and random constant strategies, we also broaden
and enhance the optimality guarantees of two approaches for learning
algorithms. The first approach involves algorithms that replicate
the asymptotic growth rate of the random optimal constant strategy,
learning solely from past assets' returns. These algorithms have the
advantage of simplicity, as they avoid the difficulties of requiring
knowledge of side market information and the latent dependence structure,
relying instead on publicly available data. In contrast, the second
approach involves algorithms that replicate the asymptotic growth
rate of the log-optimal strategy, which necessitate knowledge of the
dependence structure and complete past side market information for
estimation. These algorithms encounter challenges such as high-dimensional
market feature data and the inaccessibility of all relevant information.
Fortunately, the existence of an optimal constant strategy enables
this type of algorithm to bypass the need for side information and
reduce the memory length of past events, while still guaranteeing
optimality. However, this approach is less favorable due to its computational
complexity.\smallskip

\textbf{Connection between frameworks}. Another significant contribution
of this paper is the potential connection it establishes between existing
frameworks in the literature. First, the validity of the ``\emph{online
learning portfolio}'' framework, as surveyed in \citet{Erven2020},
is considerably strengthened. This framework investigates learning
algorithms that asymptotically approach the growth rate of the best
retrospective constant portfolios (i.e., constant strategies) using
only past assets' returns. Its justification in finance, originally
stated in the pioneering works of \citet{Cover1991,Cover1996,Cover1996a,Cover2006,Ordentlich1996},
is that these algorithms can asymptotically achieve the optimal limiting
growth rate among all dynamic strategies if the market of sole assets'
returns is i.i.d. Now, we extend this framework by demonstrating that
it guarantees the optimal limiting growth rate among all dynamic strategies
in a stationary market with partially observable side information
and latent dependence structure, while still relying solely on past
assets' returns. Additionally, the problem of learning side information,
as studied in \citet{Cover1996a,Bhatt2023}, which is potentially
unobservable and high-dimensional, is reduced to an application of
the online learning portfolio framework. Finally, a connection with
the framework of universal learning and prediction posed by \citet{Algoet1992},
utilizing mixing functionally generated portfolios, is established,
at least under the stationary condition of the market process. Besides,
it should be remarked further that the work by \citet{Cuchiero2019}
also provides a connection between the stochastic portfolio theory
and the log-optimal portfolio theory, restricted to the $1$-order
Markov process of sole assets' returns.

\pagebreak

\end{document}